%% file: main.tex
\documentclass[sn-mathphys-num,thm-restate,pdflatex]{sn-jnl}

\usepackage{journalheader}
\usepackage{tda}

\begin{document}

\abstract{
    \input{body/abstract}
}

\maketitle

\section{Introduction}\label{sec:intro}
\input{body/introduction}

\section{Background Definitions}\label{sec:preliminary}
\input{body/background}

\section{A Faithful Discretization}
\label{sec:discretization}
\input{body/discretization}

\section{Reconstruction-Based Proof of \thmref{main}}
\label{sec:reconstruction}
\input{body/reconstruction}

\input{body/other-transforms}

\section{Explicitly Building a Faithful Set}
\label{sec:build}
\input{body/build}

\section{Stability Results for Faithful Discretizations}
\label{sec:stability}
\input{body/stability}

\section{Discussion}
\label{sec:discussion}
\input{body/discussion}

\label{theend}


\input{body/backmatter}

\bibliography{references}

\begin{appendices}
    \section{Concise Versus Verbose}\label{append:whyverbose}
    \input{body/whyverbose}

    \section{Verbose Descriptors are Well-Defined}\label{append:welldef}
    \input{body/append-apd-welldef}

    \section{Proof of Simplex Count Lemma}\label{append:count}
    \input{body/append-count}

    \section{Omitted Algorithms and Proofs}
    \label{append:aux}
    \input{body/append-omitted}

    \section{Example of Building a Faithful Set}\label{append:example}
    \label{append:example}
    \input{body/example}

\end{appendices}

\end{document}

%% file: body/abstract.tex
The persistent homology transform, Betti function transform, and
Euler characteristic transform represent a shape with a multiset of
persistence diagrams, Betti functions, or Euler characteristic functions,
respectively, parameterized by the sphere of directions in the ambient space.
In this work, we give the first explicit construction of finite sets of
directions
discretizing the verbose variants of these transforms
and show that such
discretizations faithfully represent the underlying shape. Our discretization,
while exponential in the dimension of the shape, does not depend on any
restrictions on the particular immersion beyond general position, and is
stable with respect to various perturbations.

%% file: body/introduction.tex
Collections of topological descriptors are both empirically and
theoretically useful for
analyzing data and differentiating shapes~\cite{giusti2015clique,
lee2017quantifying, rizvi2017single,
lawson2019persistent, tymochko2020using,
wang2019statistical,singh2007topological,bendich2016persistent, turkes2022effectiveness}.
In \cite{turner2014persistent}, Turner et al.\ define the persistent
homology transform (PHT) and Euler characteristic (function)
transform~(ECT),
which map a geometric simplicial complex in
$\R^{d}$ to a family of persistence diagrams and Euler
characteristic functions parameterized by the directions in~$\S^{d-1}$.
They
show that, for $d=2,3$, these \emph{uncountably infinite} parameterized families
are \emph{faithful} representations of shape,
meaning that no two distinct shapes have the
same family of directional persistence diagrams or Euler characteristic functions.
A few years later, this result was proven for general $d$
in~\cite{ghrist2018euler}, by showing that injectivity of the ECT is a special
case of the inversion theorem of Schapira~\cite{schapira1995tomography},
which is a perspective more aligned with the closely related field of
tomography; see also~\cite{lebovici2022hybrid,ji2024injectivity}.
Recent related work explores both properties and variations of
directional transforms~\cite{turner2024extended, chambers2025braiding}.
In particular,
several research groups independently observed that there exist \emph{finite}
faithful representations for various types of simplicial and cubical
complexes~\cite{belton2018learning,belton2019reconstructing,betthauser2018topological,
curry2022many,micka2020searching, onus2024shoving}.

Meanwhile, researchers are already applying the PHT and the closely
related verbose PHT (VPHT) to represent various types of
data sets in machine learning and statistical pipelines~\cite{crawford2019predicting,
turner2014persistent, wect2020, betthauser2018topological,
hofer2017constructing, maria2020intrinsic}.
In many of these papers, directions are sampled uniformly, which, without
additional assumptions on the underlying shape, has no guarantee of resulting in
a faithful set. On the other hand, some theoreteical results only apply to a
sub-class of simplicial complexes and the faithful sets are large.
A main result of~\cite{curry2022many}, Theorem 7.14, bounds the number of
concise (i.e., not verbose) Euler characteristic functions or persistence diagrams
needed to form a faithful discretization, so long as the underlying simplicial
complex~$K$ satisfies additional constraints, relating to its its
``flatness'' and complexity. 
Because of the dependence on the properties of the embedding of $K$,
it is easy to identify families of simplicial complexes that make the bound
of~\cite{curry2022many} arbitrarily large. For instance, consider a path of
length two embedded in Euclidean space. By moving the vertices closer and closer
to colinear, the bound on the number of concise
persistence diagrams needed to represent the path using the methods
of~\cite{curry2022many} grows arbitrarily large; see \exref{flatness}
in \appendref{whyverbose}. Sensitivity to the
flatness of~$K$ is intimately related to using the classic
\emph{concise} versions of descriptors rather than verbose versions of
descriptors. Roughly speaking, unlike concise descriptors, verbose descriptors
do not omit the record of instantaneous changes to topological invariants
during a filtration, which has a significant impact on the behavior of faithful sets. 
As we show in \exref{flatness}, for the aforementioned family of increasingly flat
path graphs, the size of a faithful discretization made up of verbose
descriptors does \emph{not} change.  Furthermore, while results in in~\cite{curry2022many} broadly
describe conditions needed to form a finite faithful set, there are no explicit
algorithms for finding such a set, and thus, the results remain
disconnected from practical use.  

In this paper, we
provide an explicit discretization of verbose directional transforms
whose size remains constant as the shape
is perturbed, extending work that began in~\cite{belton2019reconstructing,%
fasy2022efficient}.
Belton et al.\ began in the setting of embedded plane
graphs
with~$n_0$ vertices, and gave a faithful discretization of the VPHT using
only~$\Theta(n_0^2)$ verbose persistence diagrams~\cite[Theorems 15 and
16]{belton2019reconstructing}. In
\cite[Theorem 10]{fasy2022efficient}, we remove the planar condition and
improve this bound to~$O(d + n_1 \log n_0)$, where~$n_1$ is the number of edges
in the graph. To show that these discretizations are indeed faithful,
both~\cite{belton2019reconstructing} and~\cite{fasy2022efficient} use the proof
method of \emph{reconstructing} the shape from the set of diagrams.
This work takes the natural next step of using reconstruction to give the first
explicit faithful discretization of the VPHT, the VBT, and the
VECT, for general simplicial complexes immersed in $\R^d$.
Roughly speaking, our discretization is made up of descriptors corresponding to
directions perpendicular to and ``tilted around'' each maximal simplex of an
unknown simplicial complex $K$, and has size $\dirsetSizeFull$, where $n$ is
the number of simplices in the unknown simplicial complex.  Our first main
result is \thmref{main}, which uses reconstruction to show that a set of
descriptors satisfying certain conditions is faithful. We then go on to give
explicit algorithms for constructing such a set, which culminates in our second
main result, \thmref{explicit}.  Importantly, the size of our discretization
does not depend on any restrictions on the underlying simplicial complex beyond
standard general position assumptions.

%% file: body/background.tex
We assume that the reader is familiar with homology groups (denoted $H_*$)
and their Betti numbers~(denoted $\beta_*$).
For a more complete discussion on foundational computational topology,
we refer the reader to~\cite{edelsbrunner2010computational,chazal2016structure}.

\subsection{Directional Transforms}\label{ss:tda-defs}
\input{body/background-definitions}

\subsection{Tools for Building a Faithful Discretization}\label{ss:tools}
\input{body/background-tools}

%% file: body/background-definitions.tex

In this section, we provide definitions fundamental to directional transforms.

\paragraph*{Simplices and Simplicial Complexes}
Let $k,d \in \N$.
A \emph{(geometric)~$k$-simplex} $\sigma$ is the convex hull of a set of
$k+1$ affinely
independent points in $\R^d$, denoted $\sigma=[v_0,v_1, \ldots, v_k]$.
Each of these points is called a \emph{vertex}, and we denote the vertex set
of $\sigma$ by $\verts{\sigma}$; at times,
we use $\sigma$ in place of~$\verts{\sigma}$ when using
$\verts{\sigma}$ would make equations cumbersome to read.
We call~$k$ the \emph{dimension} of~$\sigma$,
and write~$\dim(\sigma):=k$.
For another simplex $\tau$,
we say that~$\tau$ is a
\emph{face}
of $\sigma$ and $\sigma$ is a \emph{coface} of $\tau$ if~$\emptyset \neq
\verts{\tau} \subseteq \verts{\sigma}$; we denote this relation
by~$\tau \preceq \sigma$.
If~$\tau \preceq \sigma$ but $\tau \neq \sigma$, then~$\tau$ is called
a proper
face of~$\sigma$, denoted~$\tau \prec \sigma$.
If $\sigma$ is not the proper face of any simplex, we say that $\sigma$ is
\emph{maximal}.

A \emph{GP-immersed simplicial complex}~$\simComp$ is a finite set of geometric
simplices such that the vertex set of $\simComp$ is in general
position.\footnote{By \emph{general position} we mean that, for all $k$ with $1 \leq k
\leq d+1$, every subset of $V$ of size~$k$ is affinely independent.
We denote the affine space spanned by $V$ by $\aff{V}$.
Note that this is a stronger condition than the usual
notions of immersions; as long as the vertices are in general position,
self-intersections of a GP-immersed complex need not be~transverse.} We topologize
$K$ with the Alexandroff topology.  We denote the set of~$k$-simplices
in~$\simComp$ by $\simComp_k$ and the number of simplices in $\simComp_k$ by
$n_k$.  We let $n=\sum_{k\in\N} n_k$.

\paragraph*{Filtrations and Persistent Homology}
Let~$f \colon \simComp \to \R$ be such that,
for each pair of simplices~$\tau \prec \sigma \in \simComp$, we
have~$f(
\tau) \leq f(\sigma)$.  Each sublevel set $F_x:=f^{-1}(-\infty,x]$ with~$x\in
\R$
is a simplicial complex and we call~$f$ a \emph{filter} function.
The parameterized sequence of subcomplexes $\{ F_x \}_{x \in \R}$,
along with inclusions~$F_x  \subseteq F_{y}$ for all
$x \leq y$,
is the \emph{filtration of~$K$ with respect to~$f$}.
This filtration
realizes at most~$n+1$ distinct complexes: the empty set
and~$F_{f(\sigma)}$ for each~$\sigma \in \simComp$.

One filtration of interest, an \emph{index filtration},
arises when $f$ is an \emph{index
filter},
that is~$f$ is injective and $\image(f) = \{1,2,\ldots, n \}$.
Notably, in an index filtration, simplices of~$K$ are added one by one.
If $K$ is GP-immersed in $\R^d$, another filtration of interest is the
\emph{lower-star filtration}  with respect to
a direction~$\dir \in \sph^{d-1}$.
For each vertex~$v \in \simComp_0$, the height of $v$ in direction~$\dir$
is given by the dot product,~$\dir \cdot v$.  The lower-star filter function in
direction~$\dir$,  denoted~$h_\dir: \simComp \to \R$,
defines the ``height'' of each simplex in~$\simComp$
where~$h_\dir (\sigma) = \max \{\dir \cdot v ~\vert~v \in \verts{\sigma}\}$,
i.e.,~$h_{\dir}(\sigma)$
is the height of the highest vertex in $\sigma$ with respect to~$\dir$.
The lower-star filtration is the filtration of~$K$ with
respect to~$h_{\dir}$.
For~$x,y \in \R \cup \{\pm \infty\}$ such that~$x \leq y$, we
have~$h_s^{-1}(-\infty,x]=h_s^{-1}(-\infty,y]$ if and only if no vertex has
height in the interval~$(x,y]$.

Applying the homology functor to a filtration $\{F_x \}_{x \in \R}$, we obtain the
\emph{persistence~module} $\{H(F_x) \}_{x \in \R}$.
Here, we assume that homology is computed using field coefficients (e.g.,~$\Z_2$).
Then, each $H(F_x)$ is a
graded vector space $\oplus_{k \in \N} H_k(F_x)$.
For each $k \in \N$ and for each $x \leq y$,
the inclusion of simplicial complexes~$F_x  \subseteq F_{y}$ induces a linear
map~$f_{k,x,y} \colon H_k(F_x) \to H_k(F_{y})$.
Let $\beta_{k,x,y}=\beta_{k,x,y}(f)$ denote the rank of~$f_{k,x,y}$. We define
the $k$-multiplicity at $(x,y)$ as
\begin{equation}\label{eqn:multiplicity}
    \mu_k(x,y):= \begin{cases}
        \lim_{\varepsilon \to 0} \left(\beta_{k,x,y-\varepsilon} - \beta_{k, x,y} -
        \beta_{k, x-\varepsilon,y-\varepsilon} + \beta_{k, x-\varepsilon,y}\right) &\text{if } x<y  \\
        0 &\text{else}.
        \end{cases}
\end{equation}
Then, for $k \in \N$, the~\emph{$k$-dimensional
persistence diagram} is the following multiset:
\begin{equation}\label{eqn:pd-multiset}
    \stdgm{k}{f}:= \{ (x,y)^{\mu_k(x,y)}
        \}_{(x,y) \in \overline{\R}^2},
\end{equation}
where $(a,b)^m$ denotes~$m$ copies of the point $(a,b)$,
and $\overline{\R} := \R \cup \{ \pm \infty \}$.
In other words, each~$(x,y)\in \stdgm{k}{f}$ represents a
$k$-dimensional homology
generator~$\alpha$ that is born at~$x$ (that is, $[\alpha] \in H_k(F_x)$ but
$[\alpha] \notin \image(f^{x-\varepsilon,x}_k)$)
and~dies going into $y$ (that is, $y$ is the smallest index such that there
exists~$[\alpha'] \neq [\alpha] \in H_k(F_{x})$ with~$[\alpha]=[\alpha']$
in~$H_k(F_y)$).
The \emph{persistence diagram} is the graded multiset:~$\astdgm_{f}:= \sqcup_{k \in
\N} \stdgm{k}{f}$.

\paragraph*{Verbose Persistence Diagram}
Because simplices can have the same height in a filtration, it is possible
that the birth and the death of a $k$-cycle happen simultaneously, in which
case, that cycle is not represented in the persistence diagram.
However, having every simplex ``appear'' in the
persistence diagram is helpful.
Thus,
we introduce \emph{verbose persistence diagrams} (\apds).
Given a filter function $f \colon K \to \R$,
an index filter $f' \colon K \to \N$ is called
\emph{compatible with~$f$} if for all $\tau, \sigma \in K$ such that
$f(\tau) \leq f(\sigma)$, we have $f'(\tau) \leq f'(\sigma)$.

\begin{restatable}[Verbose Persistence Diagram]{definition}{apddef}\label{def:apd}
    Given a filter $f \colon \simComp \to \R$, let~$f'$ be a compatible index
    filtration. For $k \in \N$, the \emph{$k$-dimensional verbose persistence
    diagram} is the following multiset:
    \begin{equation}\label{eqn:apd-pts}
        \dgm{k}{f} := \left\{ \big( f(f'^{-1}(i)), f(f'^{-1}(j)) \big)
        \right\}_{(i,j) \in \stdgm{k}{f'}},
    \end{equation}
    where $f(\emptyset):=\infty$
    and $f'(\emptyset):=\infty$ for convenience.
    The \emph{verbose persistence diagram~(\apd)} of~$f$ is the graded multiset of all
    $k$-dimensional verbose persistence
    diagrams:~$\adgm_{f}:= \sqcup_{k \in \Z} \dgm{k}{f}$.
\end{restatable}

Because the filter function $f'$ in the definition above need not be unique, it is
not immediately clear that this definition is well-defined. A
proof that \apds are, in fact, well-defined is found in
\appendref{welldef}.

The \apd carries more information than the persistence diagram (as the height
of each simplex is a coordinate of a persistence point). Yet they do not
require extra work to compute; algorithms for computing \pds compute \apds as
an intermediate step (see, e.g., \cite[Chapter
VII.2]{edelsbrunner2010computational}).  The \apd has been used in several
contexts already.
For example, the definition of \pd in McCleary and
Patel~\cite{PHFunc} is the same as our definition of \apd. Usher
and Zhang define verbose barcodes via the lens of filtered chain
complexes in~\cite{usher2016persistent}.  This perspective is also taken in,
e.g., \cite{memoli2023ephemeral,chacholski2023decomposing}, where instantaneous
or length-zero bars are referred to as ``ephemeral.''
In other sources, ``verbose'' is sometimes replaced by the word ``augmented,''
e.g.,~\cite{fasy2022efficient,belton2019reconstructing,micka2020searching}.
We only
use verbose persistence diagrams in this paper, so we use ``diagram'' as
shorthand for ``\apd.''

\paragraph{Betti Functions}
The $k$\th Betti number of a simplicial complex $\simComp$ is the rank of the
$k$-dimensional homology group of $\simComp$, and is denoted ~$\beta_k(\simComp)
= \text{rank}(H_k(\simComp))$.~Measuring this quantity as a filtration
parameter changes and for all~$k\in \Z$ gives
rise to \emph{Betti function}.

\begin{definition}[Betti Function (BF) and Verbose
BF]\label{def:aecc}
    Given a filter~$f \colon \simComp \to \R$,
    the~\emph{$k$\th Betti function~($k$\th \bc)} is the step
    function~$\stbc_{f,k}: \R \to \Z$,
   where
    \begin{equation}\label{eqn:bc}
        \stbc_{f,k}(t) := \beta_k \left( f^{-1}(-\infty,t]\right).
    \end{equation}
    The collection of all
    Betti functions, indexed by dimension,~$\stbc_{f}:= \{ \stbc_{f,k}\}_{k \in \N}$,
    is referred to as the \emph{Betti function}.

    Suppose $f'$ is an index
    filter compatible with $f$. Let~$\sigma \in \simComp$.
    If~$\sigma$ appearing in the index filtration of $f'$ corresponds with
    $\beta_k$ increasing (respectively, decreasing), then we say $\sigma$ is
    \emph{positive} (resp., \emph{negative}) \emph{for~$\beta_k$} with respect
    to $f'$.  The set of positive (resp., negative) simplices is denoted~$K_k^+
    \subseteq K_k$
    (and~$K_{k+1}^- \subseteq K_{k+1}$).
    Then, we define the~\emph{$k$th verbose
    Betti function},~$\augbc_{f, k}: \R \to \Z^2$, as
    \begin{align*}
        \label{eqn:abc-pts}
        \augbc_{f,k}(p)
        := \bigg(\left\lvert\{\sigma \in K_k^+  \text{ s.t. } f(\sigma) \leq
        p\}\right\rvert,
                \left\lvert\{\sigma \in K_{k+1}^- \text{ s.t. } f(\sigma)
                 \leq p\}\right\rvert \bigg).
    \end{align*}
    The set of verbose Betti
    functions indexed by dimension is the \emph{verbose Betti
    function~(VBF)}, denoted~$\augbc_f$.
    We drop subscripts when clear from context.
\end{definition}

We emphasize that, just like \apds, \abcs are \emph{dimension-returning},
meaning we know how many simplices of a given dimension are added at a particular height.
Further note that the \abc represents the \bc as a parameterized count of
positive and negative $k$-simplices, which means that the \bc is known from the
\abc, making
the \bc a weaker invariant than the
\abc.
Similarly,
because the \abc is known from the \apd, we also have that
the \abc is a weaker invariant
than the \apd; the relationship between these and other descriptors is further
explored in~\cite{schenfisch2023faithful,fasy2024ordering}.

\paragraph*{Euler Characteristic Functions}
The Euler characteristic of a simplicial complex $\simComp$ is the
alternating sum of the number of simplices of different
dimensions:~$\chi(\simComp) = \sum_{i=0}^{\maxd} (-1)^i n_i$, where
we recall that~$n_i$ is the number of~$i$-dimensional simplices in $\simComp$.
Given a filtration, the Euler characteristic with respect to the filtration
parameter is known as the Euler characteristic function~(\ecc).~Formally,

\begin{definition}[Euler Characteristic Function (\ecc) and Verbose
\ecc]\label{def:aecc}
    Given a filter~$f \colon \simComp \to \R$,
    the \emph{Euler characteristic function}~$\stecc_f: \R \to \Z$ is the step
    function:
	\begin{equation}\label{eqn:ecc}
        \steccdec{p}{f} :=
                \sum_{k=0}^{\infty} (-1)^k n_k^{(p)},
    \end{equation}
    where~$n_k^{(p)}$
    is the number of $k$-simplices
    in~$f^{-1}(-\infty,p]$.

    Suppose $f'$ is an index filter compatible with $f$. Let~$\sigma \in K$.
    If $\sigma$ appearing in the index filtration of $f'$ corresponds with
    the Euler characteristic increasing
    (respectively, decreasing), then
    $\sigma$ is even- (respectively, odd-) dimensional. The set of even (resp.,
    odd) simplices is denoted by $E$ (resp., $O$).

    We define a verbose version of the ECF: the \emph{verbose Euler
    characteristic function~(\aecc)} as the function~$\augecc_f: \R \to \Z^2$,~where
    \begin{equation}\label{eqn:aecc-pts}
        \augeccdec{p}{f}
        := \bigg(\left\lvert\{\sigma \in E  \text{ s.t. } f(\sigma) \leq
        p\}\right\rvert,
                \left\lvert\{\sigma \in O \text{ s.t. } f(\sigma)
                 \leq p\}\right\rvert \bigg).
    \end{equation}
    In other words, the \aecc represents the \ecc as a parameterized count of
    even- and odd-dimensional simplices.
\end{definition}

We can read off the \ecc from the \pd by forgetting information, and the \aecc
from the \apd.  In other words, the \ecc is a weaker invariant than the \pd and
the \aecc is a weaker invariant than the \apd.

\paragraph*{Directional Transforms}
Considering descriptors corresponding to lower-star filtrations from all
possible directions, we obtain a \emph{verbose directional transform}.

\begin{definition}[Directional Transforms]\label{def:pht}
    Given a simplicial complex~$\simComp$ GP-immersed in~$\R^d$ and a descriptor
    $\gen$, a
    \emph{directional transform} of $\simComp$ is the parameterized
    set of all directional
    descriptors of type $\gen$. This is~denoted
    \begin{equation}
        \gent{K}{\sph^{d-1}}
        := \{(\dir,\gen_{h_s} )\}_{\dir \in \sph^{d-1}}.
    \end{equation}
    If $\gen$ is verbose, we may specify that $\gent{K}{\sph^{d-1}}$
    a \emph{verbose directional transform}.
    Specifically,~$\apht{K}{\sph^{d-1}}$ is the \emph{verbose
    persistent homology transform~(VPHT)}, $\abct{K}{\sph^{d-1}}$ is the
    \emph{verbose Betti transform (VBT)}, and $\aecct{K}{\sph^{d-1}}$ is the
    \emph{verbose Euler characteristic transform~(VECT)}.
\end{definition}

Given a discrete set of directions, $S \subset \S^{d-1}$,
we similarly define
\begin{equation}
    \gent{K}{S}:= \{(\dir,\gen_{h_s} )\}_{\dir \in S} \subset
    \gent{\simComp}{\S^{d-1}}.
\end{equation}
We call $\gent{K}{S}$ a \emph{discretization} of $\gent{K}{\S^{d-1}}$.
With $S$ selected wisely, the hope is that $\gent{\simComp}{S}$
earries the same information as~$\gent{\simComp}{\S^{d-1}}$.

\begin{definition}[Faithful Discretizations]\label{def:faithful}
    Given a simplicial complex~$\simComp$ GP-immersed in~$\R^d$, a discrete set
    $S \subset \S^{d-1}$, and a descriptor type~$\gen$, we say
    that~$\gent{K}{S}$ is a \emph{faithful discretization of
    $\gent{\simComp}{\S^{d-1}}$} if, for any simplicial complex~$\simComp' \neq
    \simComp$ GP-immersed in~$\R^d$, there exists~$\dir \in S$ such
    that~$\gen_{h_{\dir}}\neq \gen_{h_{\dir}'}$, where $h_s$ and $h'_s$
    are lower-star filters of~$K$ and $K'$ with respect to direction
    $s$, respectively.  That is, no other simplicial complex GP-immersed
    in~$\R^d$ has the same parameterized set of directional~descriptors.
\end{definition}

%% file: body/background-tools.tex
In this section, we introduce lemmas and tools that are used in our proofs.
We begin with a lemma that connects simplices to points in a verbose diagram.

\begin{restatable}[Simplex Count]{lemma}{simpcount}
\label{lem:count}
    Let $\simComp$ be a simplicial complex,
    $k \in \N$, \mbox{and~$c \in \R$}.
    Let~$f \colon \simComp \to \R$ be a filter function.
    Then, the~$k$-dimensional simplices of~$\simComp$ with a function value of
    $c$ are in one-to-one
    correspondence with the points in the following~multiset:
	\begin{equation}\label{eqn:isimps}
        \left\{ (a,b) \in \dgm{k}{f} \text{ s.t. } a = c \right\}
        \sqcup
        \left\{ (a,b) \in \dgm{k-1}{f} \text{ s.t. } b = c \right\}.
	\end{equation}
\end{restatable}

The proof follows from definitions, and is included in \appendref{count}.
Next, we define a structure that helps
build a geometric intuition for several of the proofs that follow.

\begin{definition}[Filtration Hyperplane]
    \label{def:filtrationhyperplanes}
    Let $\dir \in \sph^{d-1}$ and let~$c \in \R$.
    Let~$\pLine{\dir}{c}$ be the~$(d-1)$-dimensional hyperplane
    that passes through the point~$c\dir\in
    \R^d$ and is perpendicular to $\dir$.  We denote the closed half-space
    above (below) this hyperplane with respect to $\dir$ by
    $\pLineUp{\dir}{c}$ (respectively, $\pLineDown{\dir}{c}$).

    Let $V$ be a finite set of vertices in $\R^d$ and let $h_s \colon V \to \R$
	be the lower-star filter function with respect to the direction~$s$.
    The \emph{filtration hyperplanes of~$V$} are the set of hyperplanes
    $\pLines{\dir}{V} := \{ \pLine{\dir}{h_s(v)} \}_{v \in V}$.
\end{definition}

All hyperplanes in $\pLines{\dir}{V}$ are parallel to each other and
perpendicular to the direction~$\dir$.  For a simplicial complex $K$
GP-immersed in $\R^d$, because a vertex will never cause a finite death,
the births in~$\dgm{0}{h_s}$ are in one-to-one correspondence with the vertices
of $K$ by \lemref{count}.  Thus, there is a filtration hyperplane at the height
of each vertex.
Furthermore, because events in a VPD, VBF, or VECF only occur
at vertex heights, and each vertex height corresponds to at least one event, we
know all filtration hyperplanes from any of these verbose descriptors.

By observing intersections of a sufficient number of linearly independent
filtration hyperplanes, we form a grid of points of intersections:
\begin{definition}[Filtration Grid]
    \label{def:grid}
    Let $S \subset \sph^{d-1}$ be a finite set of directions such
    that~$\mathrm{span}\{S\} = \R^d$ and let $P\subset \R^d$ be a pointset. We
    define the \emph{filtration grid of~$P$ with respect to $S$} to be the grid
    of points, $A$, arising from choosing one hyperplane $\pLines{s}{P}$ for
    each $s \in S$. That is, the filtration grid is the collection of~$\lvert S
    \rvert$-way intersections of filtration hyperplanes. Note that~$\lvert S
    \rvert \geq d$, $P \subseteq A$, and $\lvert A \rvert \leq \lvert P \rvert
    ^d$.
\end{definition}

Finally, we define the specific types of directions that are the building
blocks of our faithful discretization.  We begin with directions that are
perpendicular to a simplex in a specific way.

\begin{definition}[$P$-Perpendicular]\label{def:perpendicular}
    Let $V \subset P \subset \R^d$ such that~$P$ is in general position.  We
    say that a direction~$s \in \S^{d-1}$ is \emph{$P$-perpendicular} to $V$ if
    for all~$u \in P$ and $v \in V$, we have~$\dir \cdot u  = \dir \cdot v$ if
    and only if $u \in V$.  In other words, $\dir$ is perpendicular to
    $\aff{V}$ and no other vertex in $P$ is at the same height as the vertices
    in $V$ with respect to $\dir$.
\end{definition}

The next definition describes when a direction is a slight tilt of an
initial direction that is $P$-perpendicular to some $V$, so that a specified subset $W$
(and no other points) pops above~$V \setminus W$.

\begin{definition}[$(P,V,W,s)$-Perturbation]\label{def:perturbation}
    Let $W \subset V \subset P \subset \R^d$ such that~$P$ is in general position.
    Let $\dir \in \S^{d-1}$ be $P$-perpendicular to $V$. Then, we say
    a \mbox{direction~$s_* \in \S^{d-1}$} is a~\emph{$(P,V,W,s)$-perturbation} if
    the following hold:
    \begin{enumerate}[(i)]
        \item The direction $s_*$ is $P$-perpendicular to $V\setminus
            W$.\label{stmt:perturbation-perp}\label{stmt:perturbation-firstprop}
        \item The points in $W$ are above $V \setminus W$ with respect
            to~$s_*$.\label{stmt:perturbation-popup}
        \item For all $p \in P \setminus V$, $p$ is
            strictly above (below)
            the height of $V \setminus W$ with respect to~$s_*$
            if and only if it is strictly above (below, respectively) $V$
            with respect to
            $\dir$.\label{stmt:perturbation-saveorder}\label{stmt:perturbation-lastprop}
    \end{enumerate}
\end{definition}

If $V$ or $W$ define the vertices of a simplex $\sigma$, then we may use $\sigma$ in
place of~$\verts{\sigma}$ in the definitions above.

%% file: body/discretization.tex
In this section, we give sufficient properties
for a discretization of the VPHT, VBT, or
VECT of a given simplicial complex to be faithful.

\begin{definition}[Vertex-Isolating]\label{def:vertiso}
    Given a simplicial complex $K$ GP-immersed in~$\R^d$ and given $S \subseteq \sph^{d-1}$,
    we say that the pair~$(K,S)$ is \emph{vertex-isolating} if~$S$ contains
    $d$ linearly independent directions, denoted~$\{s_1, s_2, \ldots, s_d\}$,
    and one additional direction~$\dir$ such that the points in the
    filtration grid of $K_0$ with respect to $\{s_1, s_2, \ldots, s_d\}$ are
    uniquely ordered by their heights in direction $s$.
\end{definition}

\begin{definition}[Simplex-Isolating]\label{def:simplexiso}
    Let $K$ be a simplicial complex GP-immersed in $\R^d$ and let $S \subseteq \sph^{d-1}$.
    Let $\sigma$ be a simplex with $\verts{\sigma} \subseteq \simComp_0$.
    We say that the pair~$(K,S)$ is \emph{$\sigma$-isolating}
    if~$S$ includes a direction $\dir_{\sigma}$
    such that:
    \begin{enumerate}[(a)]
        \item $\dir_{\sigma}$ is $K_0$-perpendicular to $\sigma$,
            and\label{stmt:prop:simplexiso-perp}
        \item for each $\emptyset \neq W \subsetneq V=\verts{\sigma}$, $S$ includes a direction
            that is a
            $(\simComp_0,V,W,\dir_{\sigma})$-perturbation.\label{stmt:prop:simplexiso-tiltdown}
    \end{enumerate}
    If the pair $(K, S)$ is $\sigma$-isolating for every simplex
    $\sigma \in K \setminus K_0$, we say that~$(K, S)$ is \emph{simplex-isolating}.
	See \figref{wedge-3D}.
\end{definition}

\begin{figure} \centering \centering
    \includegraphics[height=1.7in]{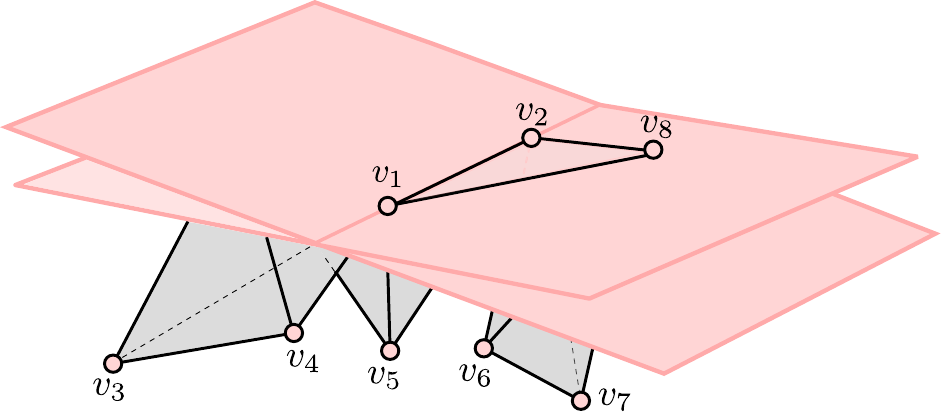} \caption[Wedge
	centered at a one-simplex] {The filtration hyperplanes (shaded in pink)
	corresponding to a pair of $[v_1, v_2, v_8]$-isolating directions,
	where we have~$V=\{v_1, v_2, v_8\}$ and $W= \{v_8\}$.  One hyperplane
	corresponds to a direction $\dir_V$ that is $\simComp_0$-perpendicular
	to $V = [v_1, v_2, v_8]$. The other corresponds to a direction that is
	a~$(\simComp_0, V, W, \dir_V)$-perturbation) which, by pivoting
	$\dir_V$ around~$V\setminus W$, ``pops'' the vertex of~$W$ above the
	filtration hyperplane.
    }\label{fig:wedge-3D}
\end{figure}

Next, we make the observation that if there are directions that
are~$\sigma$-isolating for all maximal simplices of a complex, then the set of
directions is simplex-isolating. Namely, this implies there are directions that
are $\sigma$-isolating for \emph{all} simplices of a complex, regardless if the
simplices are maximal.

\begin{lemma}[Recursive Nature of $\sigma$-Isolating Directions]\label{lem:recursive}
    Let $\simComp$ be a simplicial complex GP-immersed in~$\R^d$, and let $S$ be a
    set of directions such that $(K,S)$ is~$\sigma$-isolating for every maximal
    simplex $\sigma$.
    Then, $(K,S)$ is $\tau$-isolating for every simplex $\tau \in K$. That is,
    $(K,S)$ is simplex-isolating.
\end{lemma}

\begin{proof}
    Let $\tau \in \simComp$; we aim to show $(K,S)$ is $\tau$-isolating. 
    Either $\tau$ is a maximal simplex or $\tau$
	is a proper face of a maximal simplex.  If~$\tau$ is maximal, the claim
	follows immediately by assumption.  Suppose, then, that $\tau$ is a
	proper face of a maximal simplex $\tau_m$.  Denote the vertex sets of
	$\tau$ and~$\tau_m$ by $T$ and $T_m$, respectively.

        First, we show $S$ contains a direction that is $K_0$-perpendicular to
        $\tau$.  Because~$S$ satisfies \stmtref{prop:simplexiso-perp} of
        \defref{simplexiso}, there is some~$\dir_m \in S$ that
        is~$\simComp_0$-perpendicular to~$T_m$.  Then, because~$S$ also
        satisfies \stmtref{prop:simplexiso-tiltdown}, and, in particular,
        \mbox{when $V=T_m$} \mbox{and~$W=T_m \setminus T$}, there is a
        direction~$\dir_\tau \in S$ that is a~$(\simComp_0, T_m, T_m
        \setminus~T, \dir_m)$-perturbation; $\dir_\tau$
        is~$\simComp_0$-perpendicular to~$T_m \setminus (T_m \setminus T) = T$,
        as desired.

	Next, we show that $S$ contains a direction that is a~$(K_0, \tau,
	\tau', s_\tau)$-perturbation for every $\tau'\prec \tau$. If~$\tau$
	is a vertex, the claim is vacuously true. We therefore proceed
	assuming there exists some $\tau' \prec \tau$ and let $T' =
	\verts{\tau'}$.
    Because~$T_m \setminus (T \setminus T')
    \subset T_m$, and because~$S$ satisfies \stmtref{prop:simplexiso-tiltdown} of
    \defref{simplexiso},
    there exists a direction $\dir_p \in S$ that is a
    $(\simComp_0, T_m, T_m \setminus (T
    \setminus T'), \dir_m)$-perturbation. We show that $\dir_p$ satisfies the
    three properties of \defref{perturbation}.
    By definition,~$\dir_p$ is~$\simComp_0$-perpendicular to $T_m
    \setminus (T_m \setminus (T \setminus T') =
    T \setminus T'$. Hence, $\dir_p$ satisfies \stmtref{perturbation-perp} of \defref{perturbation}.
    Also by definition, the points
    of~$T_m \setminus (T \setminus T')$ are
    above~$T \setminus T'$ with respect to~$\dir_p$.  Because~$T' \subset T_m \setminus (T \setminus T')$,
    we conclude
    that the points in $T'$ are above~$T \setminus T'$ with
    respect to $\dir_p$.
    Hence, $\dir_p$ satisfies \stmtref{perturbation-popup} of \defref{perturbation}.
    Finally, let~$p \in \simComp_0
    \setminus T$. If~$p \in \simComp_0 \setminus T_m$,
    then $p$ is strictly above (below) the height of $\tau \setminus \tau'$ with
    respect to $\dir_p$ if and only if it is strictly above (below)
    the height of $\tau \setminus \tau'$ with respect to~$\dir_m$ by definition.
	Furthermore, this means~$p$ is above (below, respectively)~$\tau
	\setminus \tau'$ if and only if it is above (below) with
    respect to~$\dir_\tau$.  If, instead,~$p \in \simComp_0 \setminus
    (T_m \setminus T)$, then~$p \in \verts{\sigma}
    \setminus (T \setminus T')$, so $p$ is necessarily
    above~$\tau \setminus \tau'$ with respect to $\dir_\tau$ and~$\dir_p$.
    Hence,~$\dir_p$ satisfies \stmtref{perturbation-saveorder} of \defref{perturbation}
    and~$\dir_p$ is a $(\simComp_0, \tau, \tau',
    \dir_\tau)$-perturbation, as desired.
\end{proof}

We arrive at our first main theorem and related corollary:

\begin{theorem}[Sufficient Conditions for Faithful Discretization]\label{thm:main}
    Let $\simComp$ be a simplicial complex GP-immersed in~$\R^d$ such that
    $\dim(\simComp)< d$  and let~$S \subset \S^{d-1}$ be a discrete set such
    that $(\simComp,S)$ is vertex- and simplex-isolating.
    Let $\gen \in \{\adgm, \augbc, \augecc\}$.
    Then,~$\gent{\simComp}{S}$ is a
    faithful discretization of~$\gent{\simComp}{\sph^{d-1}}$.
\end{theorem}

\begin{corollary}[Subcomplexes are Represented]\label{cor:subcomplex}
    Let $\simComp$ be a simplicial complex GP-immersed in $\R^d$, let $S
    \subset \S^{d-1}$ satisfy the assumptions of \thmref{main}, and let~$L$
    be any subcomplex of $K$.  Let $\gen \in \{\adgm, \augbc, \augecc\}$.  Then
    $\gent{\simComp}{S}$ is a faithful discretization of~$\gent{L}{\sph^{d-1}}$.
\end{corollary}

The following remark emphasizes that the use of verbose descriptors, rather
than concise descriptors, is vital to our results.

\begin{remark}[Verbose Advantages]\label{rmk:verbose}
    Already, we see the dramatic advantages of using
    verbose descriptors over concise (non-verbose) descriptors.
    For example, known bounds for concise
    discretizations, e.g., given in~\cite{curry2022many}, are stated
    for restricted classes of simplicial complexes.  Changing these
    restrictions slightly leads us to constructions where this bound on the
    number of concise descriptors in a faithful representation is
    arbitrarily large, but the number of verbose descriptors described in
    \thmref{main} remains constant (and small); see~\appendref{whyverbose} for
    further~details.
\end{remark}

%% file: body/reconstruction.tex
In this section, we prove \thmref{main} (Sufficient Conditions for
Faithful Discretization) by using the discretization to reconstruct
a simplicial complex.
Our method first finds all zero-simplices, then all one-simplices, and so on.
We begin by framing all arguments with respect to the VPHT and discuss
extensions to the VBT and VECT in later~sections.

\subsection{Vertex Reconstruction}\label{ss:verts}

Vertex reconstruction using a vertex-isolating set is an established result.

\begin{lemma}[Vertex Reconstruction, {\cite[Theorem
    9]{belton2019reconstructing}}]\label{lem:vertiso}
    Let $K$ be a simplicial complex GP-immersed in $\R^d$. Then, given a set of
    directions $S$ that is vertex-isolating, $\apht{\simComp}{S}$ can
    reconstruct $K_0$.
\end{lemma}

We refer the reader to~\cite{belton2019reconstructing} for details;
simplex-isolating directions create
exactly~$n_0$~$(d+1)$-way
intersections of their filtration hyperplanes.  It is then a simple matter
of checking for the locations of these intersections.
See \figref{filtlinesgen}.
\begin{figure}[h!] \centering
    \includegraphics[width=.4\textwidth]{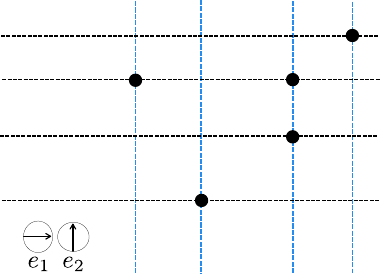}
    \hspace{1cm}
    \includegraphics[width=.4\textwidth]{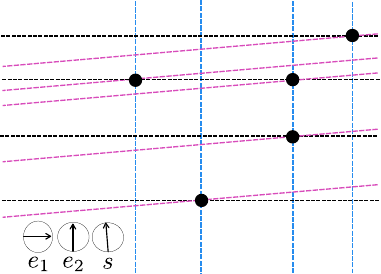}
    \caption[Vertex-Isolating Directions]{The vertex set $P$ (large black points)
    defines the filtration grid of $P$ with respect to $\{e_1, e_2\}$, denoted
    $A$ (grey and black dots in the left figure).
    Because~$e_1$ and $e_2$ are linearly
    independent and because the direction $s$ uniquely orders the points of the
    $A$, the set $\{e_1, e_2, s\}$ is vertex-isolating for the given vertex set.
    To locate the vertices of the set, we simply need to identify all
    intersections of~$\pLines{s}{P}$~(tilted pink dashed lines) with $A$.}
    \label{fig:filtlinesgen}
\end{figure}

\subsection{Higher-Dimensional Simplex Reconstruction}\label{ss:higher}
\label{sss:inkdeg}
In this section, we focus on methods for finding
higher-dimensional simplices, assuming that zero-dimensional simplices have
already been found.
\input{body/reconstruction-higher}

\subsection{Proof of \thmref{main} for the VPHT}\label{ss:full-alg}\label{ss:representation}
\input{body/reconstruction-representation}

%% file: body/reconstruction-higher.tex
The key to determining whether a simplex exists
is the \emph{$k$-indegree} of a potential simplex, $\sigma$, which
is the count of~$k$-dimensional
cofaces of $\sigma$ contained in $K$ occurring at the same height
as~$\sigma$ in a particular direction. Importantly, these cofaces need not be
proper, so for the particular case when $k = \dim(\sigma)$, we count 1 if
$\sigma \in K$, and 0 if $\sigma \not\in K$.

\begin{definition}[$k$-Indegree for Simplex]\label{def:inkdeg}
    Let $\simComp$ be a simplicial complex GP-immersed in~$\R^d$ and
    let~$\sigma \subset \R^d$ be a simplex such that $\verts{\sigma} \subseteq
    K_0$.  Furthermore, let $\dir \in \sph^{d-1}$ be a direction
    $K_0$-perpendicular to $\sigma$.  Then, the \emph{$k$-indegree of $\sigma$
    in direction $\dir$} is the number of~$k$-dimensional cofaces of $\sigma$
    contained in $K$ that have the same height as $\sigma$ in direction $\dir$.
    We say that such a coface \emph{contributes} to the $k$-indegree of
    $\sigma$ in direction $\dir$.
\end{definition}

\begin{figure}[h!] \centering \includegraphics[height=1.5in]{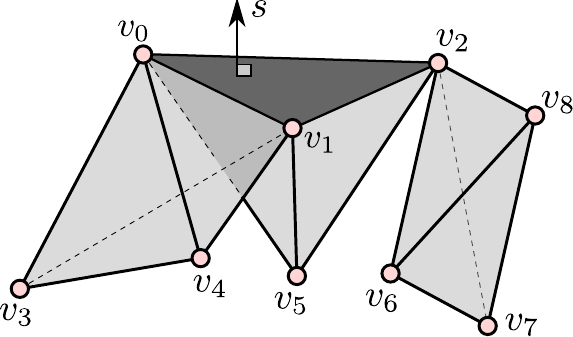}
    \caption[Computing $k$-indegree]{ Computing the three-indegree for a
    two-simplex (triangle) in $\R^4$.  The direction~$\dir \in \sph^3$ is
    orthogonal to $\aff{\sigma}$, where $\sigma=[v_0,v_1,v_2]$ (shown in dark
    gray). The three vertices $v_0$, $v_1$, and $v_2$ are at the same height
    with respect to $\dir$, and all other vertices shown are below~$\sigma$
    (note that $\dir=\dir_{\sigma}$ from \defpartref{simplexiso}{perp}).
    Although the three-indegree of~$\sigma$ is one, the VPD in direction $\dir$
    sees three tetrahedron at the same height as~$\sigma$.  We use the
    three-indegree of all faces of $\sigma$ in tilted directions (that is,
    directions given in \defpartref{simplexiso}{tiltdown}). Here, $[v_0, v_1]$
    and $[v_2]$ are the only faces with non-zero three-indegree in the tilted
    directions; they both have
    indegree values of one. Using \corref{subtract}, we subtract
    these values from the number of three-simplices
    at height $\dir \cdot \sigma$. Then, we find $3-1-1=1$, the correct value
    for the three-indegree of $\sigma$ in direction $\dir$.
    }
    \label{fig:kindegree}
\end{figure}
If a direction $\dir$ is $K_0$-perpendicular to $\aff{\sigma}$, all
zero-simplices of $\sigma$ are at the same height in direction~$\dir$.
However,
not all~$k$-simplices at this height contribute to the~$k$-indegree of
$\sigma$, as shown in \figref{kindegree}. Thus, if we only use diagrams that
are~$K_0$-perpendicular to a simplex, we may overcount $k$-indegree.  To
correctly compute~$k$-indegree, we combine
information from~$\sigma$-isolating directions (\defref{simplexiso}), using an
inclusion-exclusion type formula.
Directions that are $\sigma$-isolating are slight perturbations of $s$,
isolating faces of $\sigma$, so that we can check for $k$-simplices at the
height of $\sigma$ that are ``attached'' to $\sigma$, but that are not
cofaces of $\sigma$.

\input{body/algorithms/inkdeg}
To prove the correctness of \algref{inkdeg}, we make the following observation.

\begin{restatable}[Indegree Contributors]{lemma}{uniqueiso}
    \label{lem:uniqueiso}
    Let $\simComp$ be a simplicial complex GP-immersed in~$\R^d$.  Let $\tau$
	and $\sigma$ be potential simplices with $\verts{\tau} \subset
        \verts{\sigma} \subseteq K_0$, and \mbox{let~$\sigma'\in \simComp$}.
	Let $k$ denote $\dim(\sigma')$.  Let $\dir$ be
	$\simComp_0$-perpendicular to~$\sigma$ and let $\dir_\tau$ be
	a~$(\simComp_0, \sigma, \sigma \setminus \tau, \dir)$-perturbation.
	Then, $\sigma'$ contributes to the~$k$-indegree of~$\tau$ in
        direction~$\dir_\tau$
        if and only if~$\sigma'$ is at the same height as $\sigma$
        in direction~$\dir$ \mbox{and~$\tau = \sigma \cap \sigma'$}.
\end{restatable}

\begin{proof}
    Let $f,f_{\tau} \colon \simComp \to \R$
    be the filter function for direction $\dir$ and $\dir_\tau$, respectively.

    ($\Rightarrow$) Suppose that $\sigma'$ contributes to the $k$-indegree of
	$\tau$ in direction $\dir_\tau$. Then, by the definition of
	$k$-indegree,~$\tau \preceq \sigma'$ and $\sigma'$ is at the same height
	as $\tau$ with respect to direction $\dir_\tau$. Because $s_\tau$ is
	a~$(\simComp_0, \sigma, \sigma \setminus \tau, \dir)$-perturbation,
	this means $\sigma'$ is
	at the same height as $\tau$ in direction $s$, and therefore also at
	the same height as $\sigma$ in direction $\dir$. Because~$\tau \prec
	\sigma$ by assumption, we have~$\tau \preceq
	\sigma \cap \sigma'$. We must now show that~$\sigma \cap \sigma'
	\preceq \tau$.

    By contradiction, suppose that $\sigma \cap \sigma' \npreceq \tau$.  Then,
    there exists a vertex $v \in \sigma \cap \sigma'$
    such that $v \notin \tau$. Thus,
    $v \in \sigma \setminus {\tau}$. Because~$\dir_{\tau}$ is a
    $(\simComp_0, \sigma, \sigma \setminus \tau, \dir)$-perturbation, we know
    that~$\dir_\tau \cdot v > \dir_\tau \cdot
    w$ for all $w \in \verts{\tau}$,
    a contradiction to the claim
    that~$\sigma'$ contributes to the $k$-indegree of $\tau$
    in direction $\dir_\tau$. Therefore, $\sigma \cap \sigma' \preceq \tau$ as
    required.

    ($\Leftarrow$) Suppose that $\sigma'$ is at the same height as $\sigma$ in
	direction $\dir$ and that $\tau = \sigma \cap \sigma'$.  If~$\tau =
	\sigma'$, the claim follows from the definition of $k$-indegree.
	Then, suppose $\tau \prec \sigma'$.  Denote the dimension of $\tau$
	by $j$.  Because $\sigma'$ is a $k$-simplex,~$\tau$ is a $j$-simplex,
	and~$\tau \prec \sigma'$, we can write $\tau = [v_0, v_1, \ldots, v_j]$
	and $\sigma' = [v_0,v_1,\ldots, v_k]$ where~$v_i \in K_0$.~Then,
    \begin{equation} \label{eqn:simplexheight}
    f_\tau(\sigma') =
        \max_{i=0}^k f_\tau(v_i)
        = \max \left( \max_{i=0}^j f_\tau(v_i) , \max_{i=j+1}^k f_\tau(v_i) \right)
        = \max \left(f_\tau(\tau), \max_{i=j+1}^k f_\tau(v_i) \right).
    \end{equation}
    Because~$\sigma'$ is at the same height as $\sigma$ in direction $\dir$ and
    $\tau \prec \sigma$, $\sigma'$ is also at the same height as~$\tau$ in
    direction $\dir$, meaning that $f(v_i) \leq f(\tau)$ for all $0 \leq i \leq
    k$.
    Because~$v_i$ is not in $\tau$ for $i > j$, it must also be the case that
    $f(v_i) < f(\tau)$ for $i > j$.
    Because $\dir_\tau$ is a~$(K_0, \sigma, \sigma \setminus \tau, \dir)$-perturbation, by
    \stmtref{perturbation-saveorder} of \defref{perturbation},
    any vertex below $\tau$
    in direction $\dir$ is also below~$\tau$ in direction~$\dir_\tau$. Thus,
    $f_\tau(v_i) < f_\tau(\tau)$ for all $j < i \leq k$ and
    \begin{equation}
        \left( \max_{i=j+1}^k f_\tau(v_i) \right) < f_\tau(\tau).
    \end{equation}
    Then, by~\eqnref{simplexheight}, $f_\tau(\sigma')=f_\tau(\tau)$. This taken
    together with $\tau \prec \sigma'$ shows that~$\sigma'$ contributes to
    the~$k$-indegree
    of~$\tau$ in direction $\dir_\tau$.
\end{proof}

We make note of three points.
First, the set of $k$-simplices at the same height as~$\sigma$ is a
superset of the set of $k$-simplices that \emph{contribute} to the $k$-indegree of
$\sigma$ in this direction, i.e., those that additionally have $\sigma$ as a
face. 
Second, the contributors to the $k$-indegree of $\sigma$ also contribute to the
$k$-indegree of proper faces $\tau$ of $\sigma$, unless we consider the
$k$-indegree of $\tau$ in a tilted direction that is specifically
$K_0$-perpendicular to $\tau$ but not to $\sigma$.
And finally, if a $k$-simplex is at the same height as $\sigma$ but does not
contribute to the $k$-indegree of $\sigma$, it contributes to some proper face
of $\sigma$ instead.
Thus, in order to ensure we only count the contributors for $\sigma$, our
inclusion-exclusion type calculation of \algref{inkdeg} isolates each proper face
$\tau$ of~$\sigma$ and computes the~$k$-indegree of $\tau$ with respect to
tilted directions (\alglnref{inkdeg:dir}). 
By recursively computing the~$k$-indegree of faces of $\sigma$ in tilted
directions, we add or
subtract the number of~$k$-simplices at the same height as $\sigma$, alternating
by dimension.  This ensures that no coface of~$\tau$ that is not a coface of
$\sigma$ adds to the~$k$-indegree of $\sigma$.
\lemref{uniqueiso}
identifies which simplices contribute to the~$k$-indegree of $\sigma$, as well
as which simplices do not and, instead, contribute to the $k$-indegree of some
proper face of $\sigma$. Following the positive and negative effects of
each simplex, the calculation telescopes back to correcting the total
count simply by subtracting off any counted~$k$-simplices that do not have
$\sigma$ as a face; see \figref{kindegree}. This calculation is summarized
in the following corollary.

\begin{corollary}
    \label{cor:subtract}
    Let $\simComp$ be a simplicial complex GP-immersed in $\R^d$, let~$\sigma$
    be a potential simplex with $\verts{\sigma} \subseteq K$, and $s$ a
    direction in $\sph^{d-1}$ that is
    $K_0$-perpendicular to~$\sigma$. For $\tau \prec \sigma$, let~$\dir_\tau$
    denote a~$(K_0, {\sigma}, {\sigma \setminus \tau},
    \dir)$-perturbation, and let $\delta_\tau$ denote the $k$-indegree
    of~$\tau$ in direction $s_\tau$ for some $k \geq \dim(\sigma)$.  Finally,
    let $\delta$ denote the number of $k$-simplices born at the height of
    $\sigma$ with respect to direction $s$.  Then, the $k$-indegree of~$\sigma$
    is
    \begin{equation}\label{eqn:computekindegree}
        \delta - \sum_{\tau \prec \sigma} \inkdeg{\tau}.
    \end{equation}
\end{corollary}

We prove the correctness of \algref{inkdeg} in the following~theorem.

\begin{restatable}[Computing $k$-Indegree]{theorem}{indegValue}
\label{thm:indegValue}
    Let $\simComp$ be a simplicial complex GP-immersed in~$\R^d$.  Let $\sigma$
    be a potential simplex with $\verts{\sigma} \subseteq K_0$ and~$\dir
    \in \sph^{d-1}$ such that~$\dir$ is~$\simComp_0$-perpendicular to
    $\sigma$.  Then, \mbox{for $k \geq \dim(\sigma)$},
    $\indegAlg{\sigma}{\dir}{k}{\apht{K}{S}}$ returns the~$k$-indegree of
    $\sigma$ in direction~$\dir$.
\end{restatable}

\begin{proof}
    First, we note that because~$S$ is assumed to be simplex-isolating
    for~$K$, if $S$ is not simplex-isolating for $\sigma$, then $\sigma$ is not a
    simplex of $K$ and thus must have~$k$-indegree zero, which is the value
    returned on \alglnref{inkdeg:notasimplex}.

    Suppose then that $S$ is simplex-isolating for $\sigma$.
    We prove the claim
    inductively on~$j= \dim(\sigma)$.
	For the base case $j=0$, consider the zero-simplex $[v]$.
    Let~$h_\dir
    \colon \simComp \to \R$ be the filter function for direction $\dir$.
    We note that this base case is a  generalization
    of~\cite[Lemma~11]{belton2019reconstructing}.
    However, unlike in \cite[Lemma~11]{belton2019reconstructing},
    we are only making an argument for the $k$-indegree at a single
    vertex and not all vertices. As such, we can relax the requirement
    that no two vertices in $\simComp_0$ have the same height in
    direction $\dir$ and just require that
    no vertices in $\simComp_0 \setminus \{v\}$ have
    the same height in direction $\dir$ as $v$.
    Thus, we have that $k$-indegree of $\sigma$ is equal to the
    number of~$k$-simplices that have
    height~$h_\dir(v)$, which, by
    \lemref{count}, is:
    \begin{equation}\label{eqn:countkinvert}
        \lvert f^{-1}(f(v))\rvert = \lvert
        \{ (a,b) \in \dgm{k-1}{h_\dir} \text{ s.t. }  b = f(v)  \} \rvert
        +
        \lvert \{ (a,b) \in \dgm{k}{h_\dir} \text{ s.t. } a = f(v)  \}
        \rvert.
    \end{equation}
    In \algref{inkdeg},
    notice that if $\sigma$ is a single vertex, we do not enter the
    loop that starts on \alglnref{inkdeg:for}.  Thus, the return value is
    exactly the number given in \eqnref{countkinvert}.

    For the inductive assumption, let $j \geq 0$.
    We assume that
    \algref{inkdeg} returns the~$k'$-indegree of $\tau$ in
    direction~$\dir$, for all $\tau \in
    \simComp_j$ and all $k'\geq j$.

	For the inductive step, let $\dim(\sigma) = j+1$.  Let $k \geq j+1$.
    Now, we compute the~$k$-indegree of~$\sigma$ in direction~$\dir$. Using
\lemref{count}, we know that the number of~$k$-simplices with
    height~$h_\dir(\sigma)$ in direction $\dir$ is:
    \begin{equation}\label{eqn:numBD}
        \delta := \lvert
        \{ (a,b) \in \dgm{k-1}{h_\dir} \text{ s.t. } b = f(\sigma)  \} \rvert
        +
        \lvert
        \{ (a,b) \in \dgm{k}{h_\dir} \text{ s.t. } a = f(\sigma)  \}
        \rvert.
    \end{equation}
    Let $F_\sigma$ denote this set of simplices,
    let $\sigma' \in F_{\sigma}$,
    and let $\tau \prec \sigma$. Suppose that~$\dir_\tau$ is a~$(K_0, \verts{\sigma},
    \verts{\sigma \setminus \tau}, \dir)$-perturbation.  By \lemref{uniqueiso},
    the~$k$-simplex~$\sigma'$ contributes to the $k$-indegree of $\tau$ in
    direction $\dir_\tau$ if and only if $\tau = \sigma \cap \sigma'$.

    From \eqnref{numBD} we get $\delta$, the number of $k$-simplices born at
    the height of~$\sigma$. In \algref{inkdeg},
    \varName{numDeaths}+\varName{numBirths} is equal to
    $\delta$ in \eqnref{computekindegree}, and the values $\delta_{\tau}$ are
    computed in \alglnref{inkdeg:doublecount} of \algref{inkdeg}.  Thus, the
    return value matches \eqnref{computekindegree}, and, by \corref{subtract},
    equals the~$k$-indegree of $\sigma$.
\end{proof}

%% file: body/algorithms/inkdeg.tex
\begin{algorithm}[h!]
    \caption{$\indegAlgLong{\sigma}{\dir}{k}{\apht{K}{S}}{T = \{ \} }$}\label{alg:inkdeg}
\begin{algorithmic}[1]
	\REQUIRE $\sigma$, a simplex with $\verts{\sigma} \subseteq K_0$,
        $\dir \in S$ s.t.
	    $\dir$ is $\simComp_0$-perpendicular to $\sigma$,
	$k \in \N$ s.t. $k \geq \dim(\sigma)$,
        $\apht{K}{S}$, where $S$ satisfies the assumptions of 
        \thmref{main}, and a
        table~$T$ for memoization.
    \ENSURE the $k$-indegree for $\sigma$.
	\IF {$S$ is not simplex-isolating for $\sigma$}
	\RETURN $0$
	\label{algln:inkdeg:notasimplex}
    \ELSE
    \STATE $c \gets$ height of $\sigma$ in direction $\dir$
    \STATE $\varName{numDeaths} \gets$ number of $(k-1)$-dimensional deaths
        in $\adgm_{{\height{\dir}}}$ at
        height $c$ \label{algln:inkdeg:dk}
    \STATE $\varName{numBirths} \gets$ number of $k$-dimensional births
        in $\adgm_{\height{\dir}}$
        at height $c$\label{algln:inkdeg:dk1}
    \label{algln:inkdeg:belowcount}
    \STATE $\varName{doubleCounts} \gets 0$
    \FOR{$\tau \prec \sigma$ in non-descending order by dimension}\label{algln:inkdeg:for}
        \STATE $W \gets \verts{\sigma}\setminus \verts{\tau}$
        \STATE $\dir_{\tau} \gets$ a $(\simComp_0, \verts{\sigma}, W,
                   \dir)$-perturbation
		\label{algln:inkdeg:dir}
        \IF{$T[\tau]$ was not computed yet}
            \label{algln:inkdeg:if}
            \STATE $T[\tau]\gets
            \indegAlgLong{\tau}{\dir_\tau}{k}{\apht{K}{S}}{T}$
            \label{algln:inkdeg:doublecount}
        \ENDIF
        \STATE $\varName{doubleCounts} \gets \varName{doubleCounts} +
          T[\tau]$
        \label{algln:inkdeg:recurse}
    \ENDFOR
    \RETURN $\varName{numDeaths} + \varName{numBirths} - \varName{doubleCounts}$\label{algln:inkdeg:return}
    \ENDIF
    \label{algln:inkdeg:return}
\end{algorithmic}
\end{algorithm}

%% file: body/reconstruction-representation.tex
Using the results from the previous subsection, we arrive at \algref{full-alg}
that fully reconstructs a GP-immersed simplicial complex. Central to this
process is using $k$-indegree to certify the presence or absence of a simplex.

\input{body/algorithms/full-alg}

\begin{theorem}[Simplicial Complex Reconstruction]
\label{thm:full-reconstruction}
    Let $\simComp$ be a $\maxd$-dimensional simplicial complex GP-immersed
    in~$\R^d$, such that $\maxd \leq d-1$.
    Let $S \subset \S^{d-1}$ satisfy the assumptions of \thmref{main} (that is,
    $(\simComp,S)$ is vertex- and simplex-isolating).
    Then, \algref{full-alg} reconstructs $\simComp$, so $\apht{\simComp}{S}$ is a
    faithful discretization of $\apht{K}{\S^{d-1}}$.
\end{theorem}

\begin{proof}
    We begin by reconstructing  $K_0$ on \alglnref{full-alg:vertrecon} using the
    methods of~\cite[Theorem~9]{belton2019reconstructing}.
	\algref{full-alg} then iterates over all subsets of vertices
    $V \subseteq \simComp_0$. We do not yet know if the simplex defined by $V$ is in $K$.
    Because sets are included in non-decreasing size,~$K_{k-1}$ is finalized by
    the time $V$ is considered.
    The
    condition that the boundary of the simplex defined by $V$ is contained in
    $\simComp_{k-1}$ is checked on \alglnref{full-alg:boundary}.
    Because~$\apht{K}{S}$ is simplex-isolating, if $V$ defines a $k$-simplex
    of $K$, the set $S$
    will contain a direction $s$ that is~$K_0$-perpendicular to~${V}$.
    Thus, if there is no such direction, we know~$V$ does not define a
    simplex of $K$.
    If there is such a direction, by \thmref{indegValue},
    $\indegAlgLong{V}{\dir}{k}{\apht{K}{S}}{T = \{ \} }$ (\algref{inkdeg})
    returns the number of $k$-simplices at the height of $\aff{V}$ that
    contain the simplex defined by $V$ as a face; because $k = \lvert V
    \rvert-1$, this is either zero if $V$ does not define a $k$-simplex of
    $K$, or one if~$V$ does define a $k$-simplex of $K$. In the latter case,
    we add $V$ to $K_k$. Because we iterate over all subsets of~$K_0$, the
    algorithm eventually finds all simplices.

    Finally, because we have shown $\reconAlgo$ reconstructs $K$, we
    know~$\apht{K}{S}$ is a faithful discretization of $\apht{K}{\S^{d-1}}$.
\end{proof}

This theorem concludes the proof of \thmref{main} for the VPHT.

%% file: body/algorithms/full-alg.tex
\begin{algorithm}[!htbp]
    \caption{$\reconAlgo$}\label{alg:full-alg}
    \begin{algorithmic}[1]
        \REQUIRE $\apht{K}{S}$, where $S$ is vertex- and 
         simplex-isolating (\defsref{vertiso}{simplexiso}).
        \ENSURE simplicial complex $\simComp$.
        \STATE $K_0 \gets$ vertices of $K$, as found using the methods
        of~\cite[Theorem 9]{belton2019reconstructing}
        \label{algln:full-alg:vertrecon}
        \FOR {$V \subseteq \simComp_0$ with $1 < \lvert V \rvert \leq d$ and in
        non-decreasing size of $V$}
        \STATE $k \gets \lvert V \rvert -1$
        \label{algln:full-alg:forDimBegin}
	    \IF{ $V \setminus \{v_i\} \in \simComp_{k-1}$ for all $v_i \in V$
	    and $S$ is $V$-isolating, with $s \in S$ that is $K_0$-perpendicular
	    to $V$}
            \label{algln:full-alg:boundary}
                \IF{$\indegAlgLong{V}{\dir}{k}{\apht{K}{S}}{T = \{ \} } = 1$}
                \label{algln:full-alg:testSimp}
                        \STATE Add $V$ to $K_{k}$
                \ENDIF \label{algln:full-alg:ifEnd}
            \ENDIF
        \ENDFOR \label{algln:simComp:forDimEnd}

        \RETURN $\simComp_0 \cup \simComp_1 \cup \cdots \cup \simComp_{\maxd}$\label{algln:simComp:return}
    \end{algorithmic}
\end{algorithm}

%% file: body/other-transforms.tex
\subsection{Proof of \thmref{main} for the VBT}\label{sec:mainproof-vbt}
While \algref{full-alg} and the algorithms referenced therein are phrased for
verbose persistence diagrams, the algorithms only ever use the dimension and
heights of simplices in a given direction.%
\footnote{\label{foot:forget}In fact, if
the vertex locations are known, then we are only using dimension and the
partial order of the vertices induced by their heights in the given direction.}
That is, they never make use of the birth-death pairing information
unique to persistence diagrams.

Superficial adaptations to the algorithms of \secref{reconstruction} give us a
reconstruction method for when the necessary data of the heights and
dimensions of simplices instead comes to us from a verbose Betti function.
Thus, the same set of directions that faithfully discretize the
VPHT faithfully discretize the VBT, so we see that
\thmref{main} also holds for the VBT.

\subsection{Proof of \thmref{main} for the VECT} \label{sec:mainproof-vect}
Unlike $\apds$ or $\abcs$, $\aeccs$ do not contain information about the
dimension of each simplex, so we no
longer can directly compute $k$-indegree. However, we are still able to use a
set of directions that is vertex- and simplex-isolating to reconstruct all
simplices.
The remainder of this section serves to prove this claim.
Our main tool is an adaptation of the
definition of $k$-indegree to even/odd-indegree, where, instead of counting
simplices of dimension~$k$, we count
even- (or odd-) dimensional simplices.

\begin{definition}[Even/Odd-Indegree for Simplex]\label{def:indeg}
    Let $\simComp$ be a simplicial complex GP-immersed in~$\R^d$ and let~$\sigma
    \subset \R^d$ be a simplex such that $\verts{\sigma} \subseteq K_0$.
    \mbox{Let~$\dir \in \sph^{d-1}$} be a direction $K_0$-perpendicular to
    $\aff{\sigma}$.  Then, the \emph{even-indegree of $\sigma$ in
    direction~$\dir$}~(respectively, \emph{odd-indegree} of $\sigma$ in direction
    $\dir$) is the number of~even-dimensional~(respectively,
    odd-dimensional) cofaces of $\sigma$ contained in $K$ that have the
    same height as $\sigma$ with respect to the direction $\dir$.
    We say that such a coface \emph{contributes} to the even- (odd-)indegree of
    $\sigma$ in direction $\dir$.
\end{definition}

Just as we did with the notion of $k$-indegree, we keep track of
how many times we double-count faces of a simplex in order to correctly find the
even/odd-indegree of that simplex.  The following lemma helps us prove the
correctness of this~calculation.

\begin{lemma}[Even/Odd-Indegree Contributors]
    \label{lem:aecc_uniqueiso}
    Let $\simComp$ be a simplicial complex GP-immersed in~$\R^d$.  Let $\tau$
	and $\sigma$ be potential simplices with $\verts{\tau} \subset
	\verts{\sigma} \subseteq K_0$ and let $\sigma'\in \simComp$.  Suppose
	$\dim(\sigma')$ is even (respectively, odd).  Let $\dir$ be a direction
	that is~$\simComp_0$-perpendicular to $\sigma$ and let~$\dir_\tau$ be a
	direction that is a $(\simComp_0, \sigma, \sigma \setminus \tau,
	\dir)$-perturbation.  Then,~$\sigma'$ contributes to the~even-indegree
	(respectively, odd-indegree) of~$\tau$ in direction~$\dir_\tau$ if and
	only if $\sigma'$ is at the same height as $\sigma$ in direction $\dir$
	and~$\tau = \sigma \cap \sigma'$.
\end{lemma}

The proof is identical to that of \lemref{uniqueiso}, and we have a nearly
identical version of \corref{subtract}. We now
present \algref{evenindeg}, which computes even-indegree. The case for
computing odd-indegree is nearly identical.
\input{body/algorithms/evenindeg}
We assert the correctness of \algref{evenindeg} and its odd counterpart in the
following theorem.

\begin{restatable}[Computing Even/Odd-Indegree]{theorem}{evenindegVal}
    \label{thm:evenindegValue}
    Let $\simComp$ be a simplicial complex GP-immersed in~$\R^d$.  Let $\sigma$
    be a simplex with $\verts{\sigma} \subseteq K_0$ and~$\dir \in
    \sph^{d-1}$ such that~$\dir$ is~$\simComp_0$-perpendicular to $\sigma$.
    Then, \mbox{for $k \geq \dim(\sigma)$},
    $\evenindegAlg{\sigma}{\dir}{\aecct{K}{S}}$ returns the~even-indegree
    of $\sigma$ in direction~$\dir$.
    Similarly, the odd-version of this algorithm,
    $\oddindegAlg{\sigma}{\dir}{\aecct{K}{S}}$ returns the~odd-indegree
    of $\sigma$ in direction~$\dir$.
\end{restatable}

We provide the proof in~\appendref{omitted-vect}, as it is nearly identical to
the proof of \thmref{indegValue}, the parallel statement for $k$-indegree and
$\apds$.

The last step of the full process of reconstruction is to note that
vertex-reconstruction only relied on the presence of filtration hyperplanes,
arising from vertex-isolating directions.
Because any verbose descriptor contains complete information about filtration
hyperplanes, we can use VECFs to reconstruct $K_0$.  Finally, we give a reconstruction
algorithm and state the following main result.  Similar to
how~\algref{full-alg} uses $k$-indegree to certify the presence or absence
of a $k$-simplex, \algref{aeccfull-alg} uses even/odd-indegree to certify
the presence or absence of a $k$-simplex.
\input{body/algorithms/aeccfull_recon}
\begin{restatable}[Sufficient Conditions for Faithful Discretization of the
    VECT]{theorem}{aeccmain}\label{thm:aeccmain}
    Let $\simComp$ be a simplicial complex GP-immersed in~$\R^d$ such that
    $\dim(\simComp) < d$
    and let~$S \subset \S^{d-1}$ such that
    $(\simComp,S)$ is vertex- and simplex-isolating. Then,
    \algref{aeccfull-alg} reconstructs $K$, so $\aecct{\simComp}{S}$ is a
    faithful discretization of $\aecct{K}{\S^{d-1}}$.
\end{restatable}

Again, as the proof is nearly identical to the reasoning given in the proof
of~\thmref{full-reconstruction}, we relegate the proof
to~\appendref{omitted-vect}.  This theorem concludes the proof of \thmref{main}
for the VECT.

%% file: body/algorithms/evenindeg.tex
\begin{algorithm}[h!]
    \caption{$\evenindegAlgLong{\sigma}{\dir}{\aecct{K}{S}}{T = \{ \}
    }$}\label{alg:evenindeg}
\begin{algorithmic}[1]
    \REQUIRE $\sigma$, a simplex with $\verts{\sigma} \subseteq K_0$,
        $\dir \in S$ s.t.
            $\dir$ is $\simComp_0$-perpendicular to $\sigma$,
        $\aecct{K}{S}$, where $S$ satisfies the assumptions of 
        \thmref{main}, and a
        table~$T$ for memoization.
    \ENSURE the even-indegree for $\sigma$.
	\IF {$S$ is not simplex-isolating for $\sigma$}
	\RETURN $0$
	\label{algln:eveninkdeg:notasimplex}
    \ELSE
    \STATE $c \gets$ height of $\sigma$ in direction $\dir$
    \STATE $\varName{allEven} \gets$ number of even-dimensional events
        in $\widetilde{\chi}_{\height{\dir}}$ at
        height $c$ \label{algln:eveninkdeg:dk}
    \label{algln:eveninkdeg:belowcount}
    \STATE $\varName{doubleCounts} \gets 0$
    \FOR{$\tau \prec \sigma$ in non-descending order by
    dimension}\label{algln:eveninkdeg:for}
        \STATE $W \gets \verts{\sigma}\setminus \verts{\tau}$
        \STATE $\dir_{\tau} \gets$ a $(\simComp_0, \verts{\sigma}, W,
                   \dir)$-perturbation
		\label{algln:eveninkdeg:dir}
        \IF{$T[\tau]$ was not computed yet}
            \label{algln:eveninkdeg:if}
            \STATE $T[\tau]\gets
            \evenindegAlgLong{\tau}{\dir_\tau}{\aecct{K}{S}}{T}$
            \label{algln:eveninkdeg:doublecount}
        \ENDIF
        \STATE $\varName{doubleCounts} \gets \varName{doubleCounts} +
          T[\tau]$
        \label{algln:eveninkdeg:recurse}
    \ENDFOR
    \RETURN $\varName{allEven} - \varName{doubleCounts}$\label{algln:eveninkdeg:return}
    \ENDIF
    \label{algln:return}
\end{algorithmic}
\end{algorithm}

%% file: body/algorithms/aeccfull_recon.tex
\begin{algorithm}[h!]
    \caption{$\aeccreconAlgo$}\label{alg:aeccfull-alg}
    \begin{algorithmic}[1]
        \REQUIRE $\aecct{K}{S}$, where $S$ is vertex- and
         simplex-isolating (\defsref{vertiso}{simplexiso}).
        \ENSURE simplicial complex $\simComp$.
        \STATE $K_0 \gets$ vertices of $K$, as found using the methods
        of~\cite[Theorem 9]{belton2019reconstructing}
        \label{algln:aeccfull-alg:vertrecon}
        \FOR {$V \subseteq \simComp_0$ with $1 < \lvert V \rvert \leq d$ and in
        non-decreasing size of $V$}
        \STATE $k \gets \lvert V \rvert -1$
        \label{algln:aeccfull-alg:forDimBegin}
            \IF{ $V \setminus \{v_i\} \in \simComp_{k-1}$ for all $v_i \in V$
	    and $S$ is $V$-isolating, with $s \in S$ that is $K_0$-perpendicular
	    to $V$}
            \label{algln:aeccfull-alg:boundary}
    \IF {$\lvert V \rvert -1$ is even}
	    \STATE $x \gets \evenindegAlgLong{V}{\dir}{\aecct{K}{S}}{T = \{ \}}$
	    \COMMENT{.37in}{\algref{evenindeg}}
            \label{algln:aecccertificate:retvaleven}
    \ELSE
	    \STATE $x \gets \oddindegAlgLong{V}{\dir}{\aecct{K}{S}}{T = \{ \}}$
            \COMMENT{.37in}{Odd \algref{evenindeg}} 
            \label{algln:aecccertificate:retvalodd}
    \ENDIF
	\IF {$x \geq 1$}
            \STATE Add $V$ to $K_{k}$
	\ENDIF
                \ENDIF \label{algln:aeccfull-alg:ifEnd}
        \ENDFOR
        \RETURN $\simComp_0 \cup \simComp_1 \cup \cdots \cup \simComp_{\maxd}$\label{algln:simComp:return}
    \end{algorithmic}
\end{algorithm}

%% file: body/build.tex
In this section, given a simplicial complex $\simComp$ GP-immersed in $\R^d$,
we provide algorithms to explicitly construct a set of directions~$S \subseteq
\sph^{d-1}$ so that $(K,S)$ is vertex- and simplex-isolating; see
\appendref{example} for an example of the algorithms applied to a specific
simplicial complex.

\subsection{Auxiliary Constructions}
\label{sec:aux}
\input{body/auxiliary}

\subsection{Building an Explicit Set}
\label{sec:explicit}
\input{body/explicit}

%% file: body/auxiliary.tex
First, we describe the auxiliary methods used to compute vertex- and
simplex-isolating~directions.

\paragraph*{Computing a Perpendicular Direction}
We often compute directions orthogonal to the affine plane spanned by some set
of $k+1$ points. To do so, we use Gram--Schmidt orthogonalization to first find
$k$ vectors in the space spanned by the points, and then again perform Gram--Schmidt
orthogonalization on the standard basis vectors until we find one not in the
space spanned by the points.  While intuitive, an explicit formulation is not
easily accessible from the literature, so we provide~\algref{perp}
in~\appendref{omitted-aux}.  As a consequence of this algorithm, we get the following
lemma.

\begin{restatable}[Computing a Perpendicular Direction]{lemma}{gs}\label{lem:gs}
    Let $P \subset \R^d$ be a point set in general position with $\lvert P
    \rvert \leq d$.
    Then, we can compute a
    direction $s \in \sph^{d-1}$ orthogonal to~$\aff{P}$ and in $\perpTimeFull$
    time.
\end{restatable}

\paragraph*{Plane Filling}
Given a point set $P \subset \R^d$ of $k+1\leq d$ affinely independent points
and a direction~$\dir$ perpendicular to $\aff{P}$, we want to find enough
additional points in $\R^d$ so that they, along with the original point set
$P$, ``fill'' the~plane orthogonal to $\dir$.  We provide pseudocode for this
procedure and details of its correctness in~\appendref{omitted-aux}. A
consequence of~\algref{planefill} is the following lemma.

\begin{lemma}[Plane Filling] \label{lem:planefill}
    Let $P \subset \R^d$ be a point set in general position with $\lvert P
    \rvert \leq d$.
    Given a direction~$\dir$ perpendicular to $\aff{P}$, we can compute a
    complementary set of points $P'$ such that~$\aff{P' \cup P}$ has only two
    perpendicular directions, $\dir$ and~$-\dir$ in~$\planefillTimeFull$~time.
\end{lemma}

\paragraph*{Tilting}
    In the \algref{tilt}, we find a direction that is a slight tilt of one
input direction towards another, so that no vertex orders change.
To help explain the geometric
    intuition of the algorithm,
    let $S$ be the set of $n$ line segments in $\R^2$:
    \begin{equation}
        S := \bigg\{ \overline{(0,\dir \cdot p), (1,\dir' \cdot p)}
        \bigg\}_{p \in P}.
    \end{equation}
    Each line segment in $S$ represents a linear interpolation between the
    points~$(0,\dir \cdot p)$ and $ (1, \dir' \cdot p)$,
    which correspond to the heights in directions~$\dir$ and~$\dir'$ of each
    point in~$P$.
    \begin{figure}
        \centering
        \includegraphics[height=1.5in]{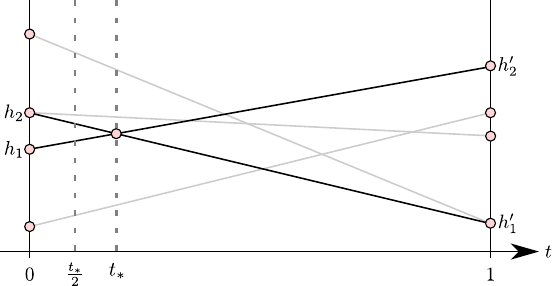}
        \caption[Geometric interpretation of \algref{tilt}] {The solid grey
        lines in the figure above indicate the changing heights of points as we
        swing direction $\dir$ towards $\dir'$. Although we do not explicitly
        compute the grey lines, we know by simple geometry that no intersection
        of grey lines (and, in particular, no swapping of point orders) occurs
        before the $t$-value $t_*$, which corresponds to the intersection of the
        closest pairwise heights of points on the left and the extremal heights
        of points on the right, as indicated by the black lines. Because there are
        no crossings of line segments before $\frac{1}{2}t_*$, there is
        therefore no change in the order of points with respect to
        direction~$\dir_t = (1-\frac{1}{2}t_*)\dir + \frac{1}{2}t_* \dir'$.
    }
    \label{fig:interpolation} \end{figure}
    We parameterize each line segment in~$S$ as~$(1-t)\dir \cdot p
    + t \dir' \cdot p$ for~$t \in [0,1]$; see the grey and black solid lines in
    \figref{interpolation}.  Then, the vertical cross sections record point
    heights with respect to some direction that is an interpolation of $\dir$
    and $\dir'$.  We want to identify some~$t_* > 0$ such that the
    ordering of the heights of points in direction~$\dir_t = (1-t_*)\dir + t_*
    \dir'$ is consistent with the ordering of the points in direction~$\dir$.
    Notice that no swapping of point order can occur before the intersection of
    black lines in \figref{interpolation}; thus, we choose the direction that is halfway to this
    point of the interpolation (labeled as $t_*/2$ in
    \figref{interpolation}) as the output of \algref{tilt}.
\input{body/algorithms/tilt}
\begin{restatable}[Tilt]{lemma}{tilt}\label{lem:tilt}
    Let $P \subset \R^d$ be a finite set. Let $\dir, \dir' \in \sph^{d-1}$.
    Then, using \algref{tilt}, we can compute a
    direction~$\dir_t=\tiltraw{\dir}{\dir'}{P}$ in~$\tiltTimeFull$ time such
    that the following properties holds for all points~$p_1,p_2 \in P$:
    \begin{enumerate}[(i)]
	\item If $p_1$ is strictly above (below) $p_2$ with respect to direction $\dir$,
	    then $p_1$ is strictly above (below, respectively) $p_2$ with
            respect to direction $\dir_t$.\label{stmt:tilt-first}
        \item If $p_1$ and $p_2$ are at the same height with respect to
            direction~$\dir$ and $p_1$ is strictly above (below) $p_2$ with respect to
            direction~$\dir'$, then $p_1$ is strictly above (respectively,
            below) $p_2$ with respect to
            direction $\dir_t$. \label{stmt:tilt-second}
        \item If $p_1$ is at the same height as $p_2$ with respect to
            both directions $\dir$ and $\dir'$, then $p_1$ and $p_2$ are at the
            same height with respect to direction
            $\dir_t$.\label{stmt:tilt-third}
    \end{enumerate}
\end{restatable}

\begin{proof}
    First, suppose all points in $P$ have the same height with respect to
    $\dir$, satisfying the if statement on \alglnref{tilt:ifbegin}.
    \stmtref{tilt-first} is trivially satisfied. Consider the lines connecting
    the height $s \cdot P$ to the heights of points with respect to $s'$. There
    are no internal crossings of segments in this arrangement, as they all
    originate at $s \cdot P$. Then, the order of points by~$s_t = \frac{1}{2}s
    + \frac{1}{2}s'$ is equivalent to the ordering of points by~$s'$, and
    \stmtsref{tilt-second}{tilt-third} are also
    satisfied.

    Next, suppose there are points in $P$ with unequal heights with respect to
    $\dir$, so on \alglnref{tilt:closest}, we let $h_1, h_2$ be the heights of
    points in $P$ with the closest unequal heights in direction $\dir$.
    Additionally, we let $h_1', h_2'$ be the heights of points in $P$ with the
    extreme heights in direction $\dir'$, as in \alglnref{tilt:extremal}.
    See \figref{interpolation}.
    Consider the lines connecting~$h_1$ to~$h'_2$ and~$h_2$ to
    $h_1'$; they intersect at some point $i$ that is at least as close to the
    left as the leftmost non-zero intersection of all linear interpolations of
    point heights with respect to directions $\dir$ and $\dir'$; see
    \figref{interpolation}. Let $t_*$ denote the first coordinate of $i$ as on
    \alglnref{tilt:tstar}.  Because $\frac{1}{2}t_* < t_*$, segments in the
    interval $(0, \frac{1}{2}t_*]$  of linear interpolations of point heights
    must not have any crossings, and so the ordering of points that have unique
    heights with respect to $\dir$ is maintained, i.e., $\dir_t$ satisfies
    \stmtref{tilt-first}.  Furthermore, if points have the same height with
    respect to direction~$\dir$, in the linear interpolation of point heights,
    this corresponds to an intersection at $t=0$, so $\dir_t$ orders points
    equivalently to how $\dir'$ orders the points, satisfying
    \stmtsref{tilt-second}{tilt-third}.

    The runtime is dominated by the else statement on \alglnref{tilt:else}.  To
    find $h_1, h_2, h_1',$ and~$h_2'$, we sort the heights of points with
    respect to directions $\dir$ and~$\dir'$ in~$\Theta(\lvert P \rvert \log
    \lvert P \rvert)$ time. Finding the intersection $t_*$ of the resulting
    two segments takes constant time, and returning the direction $\dir_t =
    (1-\frac{1}{2}t_*) \dir + \frac{1}{2}t_* \dir'$ takes $\Theta(d)$ time.
    Thus, the total runtime is~$\tiltTimeFull$.
\end{proof}

\paragraph*{Tilting to Pop}
Given a point set $P \subset \R^d$ in general position, two sets $W \subseteq V
\subseteq P$, and a direction~$\dir$ that is~$P$-perpendicular to $V$,
\algref{pivot} calculates a direction that is close to $\dir$ that ``pops'' the
points of~$W$ above the points of~$V\setminus W$.
\input{body/algorithms/pivot-down}

The algorithm begins by finding
a set of points $V' \subset \R^d$ such that
the affine space spanned by $V' \cup V$ is a $(d-1)$-dimensional subspace of $\R^d$ in
$\planefillTimeFull$ time by
\lemref{planefill}.  The additional
points help us control which way to tilt $s$.
In particular, the direction~$\dir'$ (computed on \alglnref{tiltdown-findperp}
in $\perpTimeFull$ by \lemref{gs})
is perpendicular to $\aff{V' \cup (V\setminus W) \cup \{W-s\}}$, and is the direction
towards which we can tilt in order to ``pop'' $W$ off of the
plane orthogonal to $\dir$ at height $V \setminus W$. 
Because there are two
choices for $s'$ in \alglnref{tiltdown-findperp}, the if statement in
\alglnrefTwo{tiltdown-ifstmt}{tiltdown-fixdir} ensures that the direction
is the one such that $W$ is above $V\setminus W$, and this is completed in
$\Theta(d)$ time. Finally, we return $s_*$ on \alglnref{tiltdown-tilt} using
\algref{tilt}, taking $\tiltTimeFull$ by \lemref{tilt}.

\begin{lemma}[Tilting to Pop]\label{lem:tiltdown}
	Let $P$ be a finite point set in $\R^d$ in general position.
    Let $W \subseteq V \subseteq P$
    and $\dir \in \S^{d-1}$ such that $\dir$ is $P$-perpendicular to $V$.
    Then, \algref{pivot} calculates $\pop{P}{V}{W}{\dir}$
    in $\pivotdownTimeFull{\lvert P \rvert}$ time, and the output is a $(P, V, W,
\dir)$-perturbation.
\end{lemma}

\begin{proof}
    The runtime was justified in the paragraph above detailing the algorithm.
    Recall what it means for the returned direction to be
    a~$(P, V, W, \dir)$-perturbation:
    \begin{enumerate}[(i)]
        \item The points in $W$ are above $V \setminus W$ with respect to
            direction returned.\label{stmt:tiltdown-popup}\label{stmt:tiltdown-firstprop}
        \item For all $p \in P \setminus V$, $p$ is
            strictly above (below)
            the height of $V \setminus W$ with respect to the direction returned
            if and only if it is strictly above (below, respectively) $V$
            with respect to $\dir$.\label{stmt:tiltdown-saveorder}
        \item The direction returned is $P$-perpendicular to ${V\setminus
            W}$.\label{stmt:tiltdown-perp}\label{stmt:tiltdown-lastprop}
    \end{enumerate}

    Before proving \stmtsref{tiltdown-firstprop}{tiltdown-lastprop}, we
    establish three properties of $\dir'$ returned on
    \alglnref{tiltdown-findperp}. First, by \lemref{planefill}, ${V'\cup V}$
    consists of $d$ points in general position, and $s$ is normal to
    $\aff{V'\cup V}$. Then, $\aff{V' \cup (V \setminus W) \cup \{W -s\}}$ is
    also $(d-1)$-dimensional, so $\dir'$ is automatically $P$-perpendicular to
    this space.

    Next, we show that all
    points of $W$ have the same height with respect to $\dir'$.  If~$W =
    \emptyset$, this is trivially true, so consider $w_1, w_2 \in W$. Because
    $(w_1-s), (w_2-s) \in \aff{V' \cup (V \setminus W) \cup \{ W-s \}}$, we
    have~$\dir' \cdot (w_1 - s) = \dir' \cdot (w_2 - s)$. Then, by adding
    $\dir' \cdot s$ to both sides, we obtain~$\dir' \cdot w_1 = \dir' \cdot
    w_2$, i.e., all points of $W$ are at the same height with respect to
    $\dir'$.

    Finally, we show $W \not \subset \aff{V' \cup (V \setminus W) \cup \{ W-s \}}$.
    Suppose not.  Then, because both $W$ and $W - s$ are in $\aff{V' \cup (V
    \setminus W) \cup \{ W-s \}}$, the vector $\dir$ is parallel to~$\aff{V'
    \cup (V \setminus    W) \cup \{ W-s \}}$. Furthermore, because both
    $\aff{V' \cup (V \setminus W) \cup \{ W-s \}}$ and $\aff{V' \cup V}$
    contain the same linearly independent set of $(d-1)$ points, namely $V'
    \cup V$, the planes are equal.  However, this means $s$ is parallel to
    $\aff{V' \cup V}$, contradicting the property that $\dir$ is
    $P$-perpendicular to $\aff{V' \cup V}$. Thus, we see that $\dir'$ orders
    all vertices of $W$ on the same side of $V \setminus W$; on
    \alglnref{tiltdown-fixdir}, we ensure that all vertices of $W$ are
    specifically above $V \setminus W$ with respect to $\dir'$, therefore, by
    \lemstmtref{tilt}{second}, all points of $W$ are above the points of $V
    \setminus W$ with respect to the direction $\dir_*$ returned on
    \alglnref{tiltdown-tilt} and we have shown $\dir_*$ satisfies
    \stmtref{tiltdown-popup}.

    We must also show $\dir_*$ satisfies
    \stmtrefTwo{tiltdown-saveorder}{tiltdown-perp}. Because both
    $\dir$ and $\dir'$ are perpendicular to $\aff{V \setminus W}$, the output
    $\tiltraw{\dir}{\dir'}{P \cup V'} = \dir_*$ is perpendicular to~$\aff{V\setminus W}$
    by \lemstmtref{tilt}{third}.  To see that no other vertex of $P$ is at the
    same height of vertices in $V \setminus W$ with respect to $\dir_*$,
    let $p \in P \setminus (V \setminus W)$. If $p \in P \setminus V$, because~$p$
    is strictly above or below $V \setminus W$ with respect to direction $\dir$,
    then by \lemstmtref{tilt}{first},~$p$ is strictly above or below $V
    \setminus W$ with respect to $\dir_*$, showing
    \stmtref{tiltdown-saveorder}.  If $p \in W$, then it is at the same height
    as $V \setminus W$ in direction $\dir$ and above $V \setminus W$ in
    direction~$\dir'$, thus, by \lemstmtref{tilt}{second}, $p$ is above $V
    \setminus W$ in direction $\dir_*$ and we have shown that~$\dir_*$
    is~$P$-perpendicular to $V \setminus W$, so
    \stmtref{tiltdown-perp} of the current lemma is satisfied.
\end{proof}

%% file: body/algorithms/tilt.tex
\begin{algorithm}[h!]
    \caption{$\tiltraw{\dir}{\dir'}{P}$}\label{alg:tilt}
    \begin{algorithmic}[1]
        \REQUIRE $P \subset \R$ finite and $\dir,\dir' \in \S^{d-1}$.
        \ENSURE $s_t \in \sph^{d-1}$, a direction satisfying
            \stmtsref{tilt-first}{tilt-third} of \lemref{tilt}.
            \IF{all points in $P$ have the same height with respect to $\dir$}
            \label{algln:tilt:ifbegin}
                \STATE $t_* = 1/2$ 
            \label{algln:tilt:tstarfirst}
            \ELSE
            \label{algln:tilt:else}
        \STATE $h_1, h_2 \gets$ heights of points that are closest with respect
        to $\dir$ such that $h_1 \neq h_2$ 
        \label{algln:tilt:closest}
        \STATE $h'_1, h'_2 \gets$ minimum and maximum heights of points with respect to
            $\dir'$
\label{algln:tilt:extremal}
        \STATE $t_* \gets$ the solution to $(1-t) h_1 + t h_2' = (1-t)h_2 + t h_1'$            
        \label{algln:tilt:tstar}
        \ENDIF
        \RETURN $(1-\frac{1}{2}t_*)s+ \frac{1}{2}t_* s'$
    \end{algorithmic}
\end{algorithm}

%% file: body/algorithms/pivot-down.tex
\begin{algorithm}[h!]
    \caption{$\pop{P}{V}{W}{\dir}$}\label{alg:pivot}\label{alg:pivot}
    \begin{algorithmic}[1]
        \REQUIRE $P,V,W$ point sets in $\R^d$ such that
            $P$ is in general position, $W \subseteq V \subseteq P$, and~$\lvert
            V \rvert \leq d$; $\dir$, a direction
            $P$-perpendicular to $V$.

        \ENSURE $\dir_* \in \sph^{d-1}$, a direction that is a $(P, V, W,
        \dir)$-perturbation.
        \STATE $V' \gets \planefill{V}{s}$\label{algln:tiltdown-planefill}
           \COMMENT{2.64in}{\algref{planefill}}
        \STATE $s' \gets \perpendicular{{V' \cup (V \setminus W) \cup \{W -
        \dir\}}}$
        \hfill    \COMMENT{1.781in}{\algref{perp}}
        \label{algln:tiltdown-findperp}
        \IF{$W$ is below $V \setminus W$ with respect to
                $s'$ } \label{algln:tiltdown-ifstmt}
            \STATE $s' \gets -s'$\label{algln:tiltdown-fixdir}
        \ENDIF
        \RETURN $\dir_* = \tiltraw{\dir}{\dir'}{P \cup V'}$
        \hfill  \COMMENT{2.1in}{\algref{tilt}}
             \label{algln:tiltdown-tilt}
    \end{algorithmic}
\end{algorithm}

%% file: body/explicit.tex
Next, we use the algorithms of \secref{aux} to
construct vertex- and simplex-isolating directions for a given simplicial
complex.

\subsubsection{Directions for Vertices}\label{sssec:vertdirs}
Constructing a set of directions that faithfully represents a vertex set has
been explored in previous work. By~\cite[Lemma 7]{belton2019reconstructing}, it
suffices to
construct a set of $d$ linearly independent directions, plus one additional
direction so that there are exactly $n_0$ intersections of size $d+1$ among
all associated filtration hyperplanes.
However, the construction of this final direction given in~\cite[Lemma
8]{belton2019reconstructing} requires stricter general position assumptions to
construct the set, namely, that no two vertices share any $e_i$-coordinate
for~$1 \leq i \leq d$. Here, we provide an algorithm to produce such a $(d+1)$st
direction when our pointset satisfies only the mild general position assumptions
described in the ``General Position'' paragraph of
\ssecref{tda-defs}.
\input{body/algorithms/pointiso}
The next lemma proves correctness of \algref{pointiso}.

\begin{lemma}[Correctness of \algref{pointiso}]\label{lem:pointiso}
    Let $P$ be a finite point set in~$\R^d$ in general position. Then,
    $\pointiso{P}$ returns a direction that uniquely orders the
    filtration grid of $P$ with respect to~$\{e_1, e_2, \ldots, e_d\}$
    and runs in~$\pointisoTimeFull{\vert P \vert}$
    time.
\end{lemma}

\begin{proof}
    On \alglnref{pointiso:dirassign}, we call
    \algref{tilt}, which by \lemref{tilt} takes $\Theta(\vert A\vert \log \vert
    A\vert + d)$ time. This is called $d-1$ times in the loop
    of \alglnrefRange{pointiso:forBegin}{pointiso:forEnd}; because~$\vert
    A\vert = \Theta(\vert P\vert^d)$, \algref{pointiso} takes
    $\Theta(d(\vert P \vert^d \log \vert P \vert^d + d)) =
    \pointisoTimeFull{\vert P \vert}$ total time.

    Next, we show \algref{pointiso} is correct.
    Let $\pi_i$ be the standard projection map onto the $(e_1, e_2, \ldots
    e_i$)-plane.
    As on \alglnref{pointiso:defineA}, let $A$ be the
    filtration grid of $P$ with respect to~$\{e_1, e_2, \ldots, e_d\}$
    and note that $A$ is a grid of at most $\vert P \vert^d$
    points.
    Let $j$ be the number of times the loop has been completed.
    We use the loop invariant that, in
    \alglnrefRange{pointiso:forBegin}{pointiso:forEnd}, $\dir$ totally orders
    the unique points of the image $\pi_{j+1}(A)$.
    We first show this is true before entering the loop. We
    initialize $\dir = e_1$ on \alglnref{pointiso:dirinit}.
    Thus, because~$e_1$ totally orders the points of $\pi_1(A)$, the loop
    invariant is satisfied.

    Let $\dir^i$ denote the $i$th value of $\dir$ in \algref{pointiso} (so $s^1$ is the initial
    direction defined in \alglnref{pointiso:dirinit}, $s^2$ is the direction updated by tilting towards
    $e_2$ the first time we encounter \alglnref{pointiso:forBegin}, etc. Note
    that this means $j=i-1$.)

    Suppose that the loop invariant is true going into the for loop of
    \alglnrefRange{pointiso:forBegin}{pointiso:forEnd}.
    Recall by \lemref{tilt} that $\tiltraw{\dir^{i-1}}{e_{i}}{A}$ produces a
    direction $\dir^{i}$ so that, for all $a_1, a_2 \in A$,
    \begin{enumerate}[(i)]
	\item If $a_1$ is strictly above (below) $a_2$ with respect to direction
            $\dir^{i-1}$,
	    then $a_1$ is strictly above (below, respectively) $a_2$ with
            respect to direction $\dir^{i}$.
            \label{stmt:tiltrepeat-first}
        \item If $a_1$ and $a_2$ are at the same height with respect to
            direction~$\dir^{i-1}$ and $a_1$ is strictly above (below) $a_2$ with
            respect to direction~$e_{i}$, then $a_1$ is strictly above
            (respectively, below) $a_2$ with respect to direction $\dir^{i}$.
            \label{stmt:tiltrepeat-second}
        \item If $a_1$ is at the same height as $a_2$ with respect to both
            directions $\dir^{i-1}$ and~$e_{i}$, then~$a_1$ and $a_2$ are at the
            same height with respect to direction $\dir^{i}$.
            \label{stmt:tiltrepeat-third}
    \end{enumerate}

    Because~$\dir^{i-1}$ provided a total order of $\pi_{i-1}(A)$ by assumption
    and given the statements above, we conclude that $\dir^i$ totally orders
    $\pi_i(A)$.
    Suppose that after the loop terminates, $s^d = \pointiso{P}$, totally orders
    the points of $\pi_d(A)$. Then, because~$\pi_d(A) = A$, the final direction
    totally orders the points of $A$.  Finally, by the runtime analysis, the
    loop terminates and thus, \algref{pointiso} is correct.
\end{proof}

Using the previous lemma, we are now able to construct a set of
vertex-isolating directions. This is a generalization of~\cite[Lemma 7 and
Theorem 9]{belton2019reconstructing} and we give a brief restatement of the
main idea of the proof.

\begin{lemma}[Construction of Vertex-Isolating Directions]\label{lem:vertdirs}
    Let $K \subset \R^d$ be a GP-immersed simplicial complex. Then, the basis
    directions~$e_1, e_2, \ldots, e_d$, along with the direction $\dir =
    \pointiso{K_0}$, are vertex-isolating.
\end{lemma}
\begin{proof}
    Let $A$ denote the
    filtration grid of $K_0$ with respect to $\{e_1, e_2, \ldots, e_d\}$.
    Recall that $\pLines{\dir}{K_0}$ is the set of filtration hyperplanes of
    $K_0$ with respect to $s$ (\defref{filtrationhyperplanes}).
    Because~$\dir$ orders the points
    of $A$ uniquely by \lemref{pointiso}, we know by~\cite[Lemma
    7]{belton2019reconstructing} that the vertices $K_0$ are in one-to-one
    correspondence with the points $\pLines{\dir}{K_0} \cap A$. Briefly, this
    is because $\pLines{s}{K_0}$ has a unique hyperplane passing
    through each point of $K_0$.  Then, because each point of $K_0$ lies on some point
    of $A$ , and because no hyperplane of
    $\pLines{s}{K_0}$ passes more than one point of $A$, we have $
    \pLines{\dir}{K_0} \cap A = K_0$.
\end{proof}

\subsubsection{Directions for Higher-Dimensional Simplices}\label{sssec:higherdirs}
Next, we show how auxiliary constructions of \secref{aux} can be used to
construct sets of directions that are simplex-isolating.
If a simplex $\sigma$ is less than~$(d-1)$-dimensional,
the direction returned by $\perpendicular{\sigma}$ is not guaranteed to be
$\simComp_0$-perpendicular to
$\sigma$. Thus, to ensure we have a direction
that places $\sigma$ at a unique height, we ``pop'' off any extra vertices using
$\algoName{TiltToPop}$,
returning a tilted direction that is guaranteed to be $K_0$-perpendicular to
$\sigma$.

\begin{lemma}[Construction of $K_0$-Perpendicular Directions]\label{lem:explicitperp}
    Let $\simComp$ be a simplicial complex GP-immersed in~$\R^d$.
    Let $\sigma$ be a simplex of $\simComp$.
    Furthermore, let~$\dir =
    \perpendicular{\sigma}$, let $V$ denote the set of all vertices with the
    same height as $\sigma$ in direction $\dir$, and let~$W = V \setminus \sigma$.
    Then
    $\pop{\simComp_0}{V}{W}{\dir}$ is $\simComp_0$-perpendicular to
    $\sigma$. Furthermore, this
    direction can be computed in $\explicitperpTimeFull$ time.
\end{lemma}

\begin{proof}
    By construction, $V$ is the set of all vertices with
    height $s \cdot \sigma$, meaning that~$s$ is indeed $K_0$-perpendicular to
    $V$. The other inputs to ${\tt{TiltToPop}}$ are also valid by construction.
    Then, by \lemref{tiltdown}, the
    direction returned by $\pop{\simComp_0}{V}{W}{\dir}$ is a~$(\simComp_0, V,
    W, \dir)$-perturbation.  In particular, by \stmtref{perturbation-perp} of
    \defref{perturbation}, this means the direction is
    $\simComp_0$-perpendicular to~$V \setminus W = \sigma$, as desired.

    We compute $\dir = \perpendicular{\sigma}$ in $\perpTimeFull$ time by
    \lemref{gs}. We find $V$ and $W$ in~$O(n_0)$ time.  Finally, we compute
    $\pop{\simComp_0}{V}{W}{\dir}$.  The algorithm $\algoName{TiltToPop}$ has
    runtime~$\pivotdownTimeFull{(n_0+\lvert W \rvert)}$ by \lemref{tiltdown}.
    Because~$W$ has $O(d)$ vertices, all
    these operations in total take $\explicitperpTimeFull$~time.
\end{proof}

Putting the pieces together, we use the algorithms of this section and
\secref{aux} to construct a set of directions.

\begin{construction}[Constructing a Faithful Set]\label{con:explicit}
    Let $\simComp$ be a simplicial
    complex GP-immersed in~$\R^d$ such that $\dim(\simComp) < d$,  and
    let $S$ be the set of directions constructed iteratively as
    follows:
    \begin{enumerate}[Step 1]
        \item Initially, let $S$ be the standard basis vectors ($e_1$, $e_2$,
            \ldots, $e_d$), plus the additional direction $\pointiso{K_0}$.
            \label{step:verts}
        \item For every maximal $\sigma \in \simComp$,
            \label{step:higher}
            \begin{enumerate}[a. ref=\theenumi{} a.]
                \item Let $\dir = \perpendicular{\sigma}$, let
                    $V$ be the set of all vertices with the same height as
                    $\sigma$ in direction~$\dir$, and let
                    $W = V \setminus \sigma$.
                    Then, add the direction~$\dir_{\sigma}:=
                    \pop{\simComp_0}{V}{W}{\dir}$ to
                    $S$.\label{step:substep:perp}
                \item For each $\tau \in \simComp$ such that $\tau \prec
                    \sigma$, add the direction
                    $\pop{\simComp_0}{\sigma}{\tau}{\dir_{\sigma}}$ to $S$.
                    \label{step:substep:pivot}
            \end{enumerate}
    \end{enumerate}
\end{construction}

Finally, we show that \conref{explicit} forms a faithful discretization
of the VPHT, the VBT, and the VECT,
and analyze the size and time complexity.

\begin{restatable}[Explicit Faithful
    Discretization]{theorem}{explicit}\label{thm:explicit}
    Let $\simComp$ be a simplicial
    complex GP-immersed in~$\R^d$ such that $\dim(\simComp) < d$,  and
    let $S$ be the set of directions from \conref{explicit}.
    Let $\gen \in \{\adgm, \augbc, \augecc\}$.
    Then, $\gent{\simComp}{S}$ is a
    faithful discretization of~$\gent{\simComp}{\sph^{d-1}}$
    of size
    $\dirsetSizeFull$, and $S$
    can be computed in~$\dirsetTimeFull$ time.
\end{restatable}

\begin{proof}
    By \lemref{pointiso}, the directions added to $S$ in \stepref{verts} are
    vertex-isolating (\defref{vertiso}).
    By \lemref{explicitperp}, the directions added to $S$ in
    \substepref{higher}{perp} are $K_0$-perpendicular to every maximal simplex,
    satisfying the first condition of being simplex-isolating
    (\defpartref{simplexiso}{perp}).  By \lemref{tiltdown}, the directions
    added to $S$ in \substepref{higher}{pivot} are~$(P, V, W,
    \dir)$-perturbations, satisfying the second condition of being
    simplex-isolating (\defpartref{simplexiso}{tiltdown}).  Thus, the
    directions of \stepref{higher} are simplex-isolating.  Because~$(K,S)$ is
    both vertex- and simplex-isolating, by \thmref{main}, $\gent{K}{S}$ is a faithful
    discretization of~$\gent{K}{\sph^{d-1}}$.

    Now, we analyze size and time bounds.
    By \lemref{pointiso}, the $\Theta(d)$ directions added to $S$ in
    \stepref{verts} can be computed in
    time $\pointisoTimeFull{n_0}$.

    Next, we give bounds for \substepref{higher}{perp}.
    By \lemref{explicitperp}, for each
    maximal~$\sigma \in \simComp$, the direction $\dir_{\sigma}$ in
    \substepref{higher}{perp} can be computed in time~$\explicitperpTimeFull$.
    Because the total number of maximal simplices is $O(n)$, the total time
    computing directions in \substepref{higher}{perp} is
    $O\left(n(\explicitperpTime)\right)$.

    Given $\dir_{\sigma}$ and $\tau \prec \sigma$, by \lemref{tiltdown}, a
    single direction in \substepref{higher}{pivot} can be computed in
    time~$\pivotdownTimeFull{n_0}$.  Because, for every maximal $i$-simplex of
    $\simComp$, we compute one direction for each of its proper faces, each
    $i$-simplex adds a total of~$2^{i+1}-2$ directions in
    \substepref{higher}{pivot}. Then, the number of
    directions in \substepref{higher}{pivot} is $\Theta(n2^{\maxd})$.  Hence,
    the total time to compute directions in
	\substepref{higher}{pivot} is~$O(n2^\maxd(\pivotdownTime{n_0}))$.

    Thus, the set $S$ has $\dirsetSizeFull$ directions and can be computed in
    time~$\dirsetTimeFull$.
    \camera{time could use a double check -- this is updated after Brittany/Anna
    discussion}
\end{proof}

%% file: body/algorithms/pointiso.tex
\begin{algorithm}[!htbp]
    \caption{$\pointiso{P}$}\label{alg:pointiso}
    \begin{algorithmic}[1]
        \REQUIRE $P \subset \R^d$, a point set in general position.
        \ENSURE $\dir \in \sph^{d-1}$, a direction that uniquely orders 
         the filtration grid of $P$ with respect to $\{e_1, e_2, \ldots, e_d\}$
         \STATE $A \gets$ the filtration grid of $P$ with respect to $\{e_1,
         e_2, \ldots, e_d\}$.
         \label{algln:pointiso:defineA}
        \STATE $\dir \gets e_1$
        \label{algln:pointiso:dirinit}
        \FOR {$i$ from $2$ to $d$}
        \label{algln:pointiso:forBegin}
        \STATE $\dir \gets \tiltraw{\dir}{e_i}{A}$
        \label{algln:pointiso:dirassign}
        \ENDFOR \label{algln:pointiso:forEnd}
        \RETURN $\dir$\label{algln:pointiso:return}
    \end{algorithmic}
\end{algorithm}

%% file: body/stability.tex
In this section, we investigate the stability of our
discretization under perturbations of directions and vertices.
Central to these results is the
stratification of the sphere of directions induced by directional transforms,
as has been noted in related work, e.g.,~\cite{leygonie2019framework,
curry2022many, fasy2022persistent, arya2024sheaf,
hickok2023topics}.%
\footnote{The stability of the entire PHT follows from the stability of
(directional) persistence diagrams under continuous changes to a shape or
changing direction.
See, e.g.,~\cite[Theorem 4.11]{arya2024sheaf}.}
The stratification is defined by dividing the sphere of
directions into regions (strata) where all directions within the same stratum
induce the same partial order on vertices (ordered by their height with respect
to that direction); see \figref{strat}.

\begin{figure}[h!] \centering \includegraphics[width=.6\textwidth]{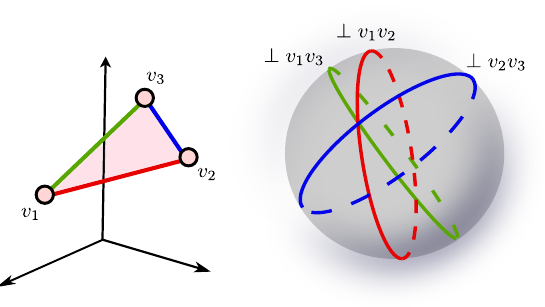}
    \caption{Three vertices in $\R^3$ (left) divide
     $\sph^2$ into $14$ strata. Each stratum is a region containing all directions that
     define the same partial order on vertices. Note that the set of
     directions perpendicular to the line defined by two vertices forms a great circle;
     the two directions perpendicular to $[v_1,v_2,v_3]$ correspond to the
     two three-way intersections of these great circles.
    }
    \label{fig:strat}
\end{figure}

We begin by defining a
strata-preserving map, which allows for a clean way to describe important
relationships between sets of directions.

\begin{definition}[Strata-Preserving Map]
    Given $S \subseteq \S^{d-1}$ and a simplicial complex~$\simComp$
    GP-immersed in~$\R^d$, the map $f: \S^{d-1} \to \S^{d-1}$ is
    \emph{strata-preserving with respect to~$K$ and $S$} if, for
    every~$\dir \in S$, the partial orderings of vertices induced by their
    heights with respect to directions $\dir$ and~$f(\dir)$ are the~same.
\end{definition}

Next, we give conditions under which a set of directions may be perturbed and
still correspond to a faithful discretization.
\begin{theorem}[Directional Perturbations]\label{thm:dirpert}
    Let $\gen \in \{\adgm, \augbc, \augecc\}$,
    let~$\simComp$ be a simplicial complex GP-immersed in~$\R^d$, and let~$S
    \subseteq \S^{d-1}$ be such that~$(K,S)$ is vertex- and simplex-isolating.
    Let~$f \colon \S^{d-1} \to \S^{d-1}$ be strata-preserving with respect
    to~$K$ and $S$.
    If~$(K, f(S))$ is vertex-isolating, then~$\gent{\simComp}{f(S)}$ is a
    faithful discretization of~$\gent{\simComp}{\sph^{d-1}}$.
\end{theorem}

\begin{proof}
    By hypothesis, the set $(K, f(S))$ is vertex-isolating.
    Let $\sigma$ be a maximal simplex.  Because~$(K,S)$ is
    simplex-isolating, $S$ contains a direction $\dir_{\sigma}$ that
    is~$K_0$-perpendicular to~$\sigma$. Because~$f$
    is strata-preserving, it preserves vertex
    order, so
    $f(\dir_{\sigma})$ is also $K_0$-perpendicular to $\sigma$.  Moreover, for
    $\tau \prec \sigma$ and for
    a~$(K_0,\sigma,\sigma \setminus \tau,s_{\sigma})$-perturbation $s' \in S$,
    the direction $f(s')$ is
    a~$(K_0,\sigma,\sigma \setminus \tau,f(s_{\sigma}))$-perturbation.
    Hence,~$(K,f(S))$ is
    $\sigma$-isolating. Because $\sigma$ was arbitrarily chosen,
    we know that~$(K,f(S))$ is simplex-isolating.
    Because~$(\simComp,f(S))$ is vertex- and simplex-isolating, by
    \thmref{main},~$\gent{K}{f(S)}$ is a faithful discretization
    of~$\gent{K}{\sph^{d-1}}$.
\end{proof}

Given a simplicial complex $K$ immersed in $\R^d$ and~$\epsilon>0$, we call $L$
an~\emph{$\epsilon$-perturbation of $K$} if there exists a simplicial
homeomorphism $g \colon K \to L$ such that, for all~$v \in K_0$, we have
$\lvert v-g(v)\rvert\leq \epsilon$.  Vertex heights change continuously when
moving a vertex continuously; the next result makes use of this fact.

\begin{theorem}[Perturbations of Immersion]\label{thm:perturb}
    Let $\gen \in \{\adgm, \augbc, \augecc\}$,
    let $\simComp$ be a simplicial complex GP-immersed in~$\R^d$,
    and let $S$ be as in \conref{explicit}.
    Let $L$ be an $\epsilon$-perturbation of $K$ such that,
    with respect to every~$\dir \in S$, the order of vertex heights in~$L_0$ with
    is the same order of vertex heights in~$\simComp_0$.
    Then, for $\epsilon$ sufficiently small,~$\gent{L}{S}$ is a
    faithful discretization of $\gent{L}{\sph^{d-1}}$.
\end{theorem}

\begin{proof}
Because $L$ is obtained from $K$ by perturbing the vertices, we have a
    homeomorphism and simplicial map $g \colon K \to L$ such that, for a vertex
    $v \in K_0$, $g(v)$ is the perturbation of $v$.
    Let $\sigma \in K$ be a maximal simplex.  Because~$(K,S)$ is
    $\sigma$-isolating, there is a direction~$\dir_{\sigma} \in S$ that
    is~$K_0$-perpendicular to $\sigma$.  Because~$K$ and $L$ have the same vertex
    order with respect to direction~$\dir_{\sigma}$,
    we know~$\dir_{\sigma}$ is $L_0$-perpendicular to~$g(\sigma)$.  Moreover,
    for $\tau \prec \sigma$, if $s'$
    is a $(K_0,\sigma,\sigma \setminus \tau,s_{\sigma})$-perturbation, then $s'$
    is also
    a~$(g(K_0),g(\sigma),g(\sigma \setminus \tau),s_{\sigma})$-perturbation.
    Hence,~$(g(K),S)=(L,S)$ is~$g(\sigma)$-isolating.  Because~$g$ is a
    bijection,~$(L,S)$ is simplex-isolating.

    Next, suppose that $s_1, s_2, \ldots, s_{d+1} \in S$ are vertex-isolating
    directions for $K$; in particular, $s_{d+1}$ orders $K_0$ uniquely.
    Consider a linear interpolation between~$K_0$ and~$L_0$ induced by $g$, and thus,
    a continuous interpolation between the arrangement of
    filtration hyperplanes of $K$ and $L$ with respect to the
    directions~$\{s_1, s_2, \ldots, s_{d+1}\}$. If $\epsilon$ is small enough,
    no topological changes occur in the interpolation of these arrangements,
    i.e., $s_{d+1}$ still orders the filtration grid of $L$ with respect to the
    directions~$\{s_1, s_2, \ldots, s_d\}$ uniquely, so $(L,S)$ is
    vertex-isolating. See \figref{grid_perturbation} for a potential issue that
    arises if we do not choose $\epsilon$ small enough.

    Thus, because $(L,S)$ is vertex- and simplex-isolating, by \thmref{main},
    $\gent{L}{S}$ faithfully discretizes $\gent{L}{\sph^{d-1}}$ as desired.
\end{proof}

\begin{figure}[h!] \centering
    \includegraphics[width=.3\textwidth]{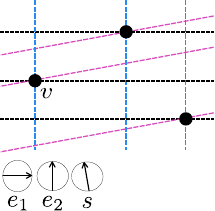}
    \hspace{1cm}
    \includegraphics[width=.3\textwidth]{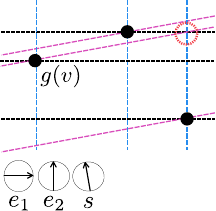}
    \caption{The direction $s$ uniquely orders the
    filtration grid of the leftmost vertex set
    with respect to~$\{e_1, e_2\}$, meaning $\{e_1, e_2, s\}$ is
    vertex-isolating.  On the right, we see that even an
    $\epsilon$-perturbation that preserves vertex orders in each direction does
    not guarantee that the set remains vertex-isolating; the new position of
    $g(v)$ causes an extra intersection (in the red dotted
    circle), so $s$ no longer uniquely orders the filtration grid and~$\{e_1,
    e_2, s\}$ is no longer vertex-isolating.  However,
    for small enough $\epsilon$, such extra
    intersections are avoided.
    }
    \label{fig:grid_perturbation}
\end{figure}

%% file: body/discussion.tex
In this work, we provide sufficient conditions for a faithful discretization of
the verbose persistent homology transform.  Because only the presence and
dimension of filtration events are used (and not birth/death pairing
information), the techniques presented in the paper immediately apply to
the verbose Betti function transform. In fact, by the same reasoning, they
apply to \emph{any} verbose directional transform that contains information
about the dimension and height of simplices.  Additionally, we show that, under
mild adaptations, the methods of this paper also yield a faithful
discretization of the verbose Euler Characteristic function transform.  This is
particularly important because, in practice, \aeccs are often preferred to
\apds due to the existence of faster algorithms for computing the functions. 

For a simplicial complex GP-immersed in~$\R^d$ with $n_0$ vertices, $n$
simplices, and dimension $\maxd$, we compute a faithful discretization of size
$\dirsetSizeFull$ in~$\dirsetTimeFull$ time.  We go on to show that the
discretization is stable with respect to multiple types of perturbations.  In
particular, the size of our discretization remains unchanged by such
perturbations and is independent of any restrictions on the simplicial complex
beyond standard general position assumptions.

The properties and explicit constructions of faithful discretizations computed
in this paper are an important step for use in both applications and theory.
While our use of verbose descriptors results in a generally much smaller
discretization than the best-known bounds for discretizations that use concise
descriptors, the sets discussed in this paper are often still much larger than is
strictly necessary to form a faithful discretization. In ongoing work, we hope
to better understand computing faithful discretizations with minimal
cardinality.

%% file: body/backmatter.tex
\backmatter

\bmhead{Acknowledgements}
This material is based upon work supported by the National Science Foundation
under the following grants: CCF 1618605 (BTF, SM), CCF 2046730 (BTF), DBI 1661530 (DLM,
LW), DGE 1649608 (AS),  DMS~1664858 (BTF, AS), and DMS~1854336
(BTF). AS is also supported by the Dutch Research Council (NWO) under project no. P21-13.
All authors thank the Computational Topology and Geometry (CompTaG)
research group at Montana State University as well as
the students in the 2023 Eidgen\"ossische Technische Hochschule Z\"urich course,
``Projects in TDA -- Directional Transforms''
for their critical reading,
thoughtful discussions, and helpful feedback.
We also thank the reviewers for
their careful reading and constructive comments that helped improve this paper.

\section*{Declarations}

\bmhead{Competing Interests} The authors have no competing interests to declare
that are relevant to the content of this article.

%% file: body/whyverbose.tex
In this appendix, we briefly make contact with the work of Curry, Mukherjee, and
Turner~\cite{curry2022many},
which explores the existence of a faithful discretization of the \emph{concise}
PHT and ECT. In particular, in Theorem 7.14, we see a bound on the size of a
faithful discretization of concise Euler characteristic functions or concise
persistence diagrams when, where assumptions about
the underlying simplicial complex must be made.

Roughly speaking, the bound of~\cite{curry2022many} is
stated for simplicial complexes that are known to satisfy
two constraints ensuring $K$ is
not too ``flat,'' or that no collection of points are too close to
coplanar.
First, we give a notion of criticality, which is restated
from~\cite[Def.~7.2]{curry2022many}.
\begin{definition}[Observable]
    Let $K$ be a simplicial complex embedded in $\R^d$, and consider a
    direction $s \in \sph^{d-1}$. We say that a vertex $v \in K_0$ is
    \emph{observable by $s$} if the concise Euler characteristic function
    changes at $s \cdot v$.
\end{definition}
Following \cite[Def.~7.5]{curry2022many}),
given a dimension $d$ and parameters $\delta\in \R_+$ and~$k \in
\Z_+$, we define a family of simplicial complexes.

\begin{definition}[Well-Behaved Families]
    \label{def:nottooflat}
    Let $\mathcal{K}(d,\delta,k)$ denote the set of simplicial
complexes~$K$~satisfying:
\begin{enumerate}
    \item $K$ is embedded in $\R^d$.
        \label{stmt:embedded}
    \item For each $v \in K_0$, there exists a ball $B \subset \S^{d-1}$ of
        radius~$\delta$ such that for all $s \in B$,~$v$ is observable by $s$.
        \label{stmt:observable}
    \item For each $s \in \sph^{d-1}$, consider the ball $B_s \subset \S^{d-1}$
        of radius $\delta$ centered at $s$. Then, the number of vertices
        observable by a direction $s' \in B_s$ is never more than $k$.
        \label{stmt:criticalitybound}
\end{enumerate}
\end{definition}
Then, we have the
following bound on the size of a concise discretization.

\begin{theorem}\cite[Thm. 7.14]{curry2022many}
    Given  a dimension $d$ and two parameters $\delta\in \R_+$ and~$k \in
    \Z_{+}$, for any shape in
    $\mathcal{K}(d, \delta, k)$, there exists a set of
    \begin{equation}
        \label{eqn:cmtbound}
        ((d-1)k
         +1) \bigg(1 +
        \frac{3}{\delta}\bigg)^d + O\bigg(\frac{d^{d+1}k^{2d}}{\delta
        ^{2d(d-1)}}
        \bigg)
    \end{equation}
    directions whose discretization of the ECT or PHT faithfully represents the
    shape, when restricted to shapes in $\mathcal{K}(d, \delta,
    k)$.
\end{theorem}

Although the above theorem requires knowing or assuming $\delta$ and $k$ a priori,
and thus takes a rather different approach than our current results, we still
see an interesting consequence of using concise versus verbose descriptors
through comparison. 
The following example highlights perhaps the most
important difference of our results: sensitivity versus insensitivity to the
``flatness'' of~$K$.

\begin{example}[Sensitivity/Insensitivity to Flatness]\label{ex:flatness}
    Consider the simplicial complex~$K$ consisting of two adjacent edges in
    $\mathbb{R}^2$, $[v_0,v_1]$, and $[v_1, v_2]$ (this is a simplification
    of~\cite[Ex. 7.4]{curry2022many}).  Denote the angle between the
    edges by $\theta$. Then, $v_1$ is only observable by directions that
    order it first or last. Such directions define a ball $B$ with radius $\delta
    = \pi - \theta$.
    Because~$v_0$ and $v_2$ are observable in balls with larger radius, we see that
    we must use $\delta \leq \pi - \theta$ as our parameter in \stmtref{observable}
    of \defref{nottooflat}.

    Because the Euler characteristic can only change at a vertex height, and all
    three vertices are indeed observable for some directions, we must use $k \geq 3$
    in \stmtref{criticalitybound}.
    Then, the bound given in~\cite{curry2022many}
    becomes
    \begin{equation}
        \begin{split}
        ((2-1)k +1) & \bigg(1 +
        \frac{3}{(\pi - \theta)}\bigg)^2 + O\bigg(\frac{2^{3}k^{4}}{(\pi - \theta)
        ^{4(1)}}
        \bigg)\\
        & = (k+1) \bigg(1 +
        \frac{3}{(\pi-\theta)}\bigg)^2 + O\bigg(\frac{k^4}{(\pi-\theta)
        ^{4}}\bigg).
        \label{eqn:elbowbound}
        \end{split}
    \end{equation}
    By moving $v_0$, $v_1$, and $v_2$
    closer to colinear, $\theta$ approaches $\pi$. Notice this does not change the bound $k
    \geq 3$, because, as long as the vertices are not quite colinear, there is
    always some direction that sees three changes to the Euler characteristic.
    Thus, simply by adjusting vertex locations slightly, we can construct
    examples where the value of \eqnref{elbowbound} is as large as we would
    like.

    On the other hand, by \thmref{main}, we only require $O(n2^{\maxd} + d)$
    directions when using verbose descriptors. Specifically, this is $d+1= 3$
    vertex-isolating directions and six simplex-isolating directions (for each
    edge, one perpendicular direction and two perturbed directions).  Thus, our
    total number of required directions is independent of the
    angle at $v_1$ and remains constant at nine.
\end{example}

Note that the example above is related to the stable behavior for verbose
descriptors described in~\thmref{perturb}.
While the bound of~\cite{curry2022many} is not shown to be tight,
sensitivity to the structure of $K$ when using concise descriptors is
unavoidable.
Despite their trivial lifespans, the ``on-diagonal'' points of
a verbose persistence diagram (or, the analogue of such points in a verbose
Euler characteristic function) have a significant impact on the behavior of
faithful sets.
A thorough exploration of the differences between concise and
verbose descriptors is outside of the scope of this paper, so we refer the
reader to~\cite{fasy2024ordering, fasy2018challenges, zhou2023beyond,memoli2023ephemeral}.
In particular,
sensitivity of concise descriptors to the structure of $K$ is highlighted
in~\cite[Section 6]{fasy2024ordering}, which shows, if we do not restrict~$\delta$ or
$k$, faithful sets of concise descriptors require multiple directions
perpendicular to each maximal simplex.

%% file: body/append-apd-welldef.tex
We begin by showing \apds are well-defined.
Existence of \apds compatible with a \pd follows from the fact that any partial
order (namely, the partial order induced by a general filtration) admits a
compatible linear extension (corresponding to the total order of a compatible
index filtration).
We therefore focus on showing uniqueness of the corresponding \apd.
This boils down
to proving that \defref{apd} is independent of the choice of the index
filtration.
To prove this claim, we will make use of the following definition, which
organizes collections of $f$-compatible filter functions.

\begin{definition}[Compatible Transposition Sequence]
    Let $\simComp$ be a simplicial complex and let $f \colon K \to \R$
    a filter function. Let $f_1, f_2, \ldots, f_m$ be $f$-compatible index
    filtrations, where each adjacent pair $f_j$ and $f_{j+1}$ differ only by a
    single transposition. That is, for some $i \in \N$, we have the equalities
	$f_j^{-1}(i) = f_{j+1}^{-1}(i+1)$ and~$f_j^{-1}(i+1) =
	f_{j+1}^{-1}(i)$, but for every other $i' \not\in \{i,i+1\}$, we
	have~$f_j^{-1}(i') = f_{j}^{-1}(i')$. This sequence of such index
	filtrations is called an~\emph{$f$-compatible
    transposition sequence}.
\end{definition}

We begin by considering transposition sequences with just two elements.

\begin{lemma}[Commutative Diagrams]
    \label{lem:complexcommute}
    Let $\simComp$ be a simplicial complex and~$f \colon K \to \R$ be a filter
    function.  Let $\{f_1, f_2\}$ be an $f$-compatible transposition
    sequence, where the transposition between $f_1$ and $f_2$ occurs at indices
    $i$ and $i+1$.     Let~$\{K_j^{(1)}\}_{j \in \N}$ and~$\{K_j^{(2)}\}_{j \in
    \N}$ denote the corresponding filtrations. Then,
    the following diagram of inclusions commutes:
    \begin{equation}\label{eqn:bigcommutedgmcplx}
        {\tiny{
        \begin{tikzcd}
            & & & K_i^{(1)}  \arrow[hookrightarrow, rd] & & &\\
            \cdots \arrow[hookrightarrow, r] &
            \begin{tabular}{@{}c@{}c@{}} $K_{i-2}^{(1)}$ \\ \verteq \\ $K_{i-2}^{(2)}$\end{tabular}
            \arrow[hookrightarrow, r] &
            \begin{tabular}{@{}c@{}c@{}} $K_{i-1}^{(1)}$ \\ \verteq \\ $K_{i-1}^{(2)}$\end{tabular}
            \arrow[hookrightarrow, ru] \arrow[hookrightarrow, rd] &  &
            \begin{tabular}{@{}c@{}c@{}} $K_{i+1}^{(1)}$ \\ \verteq \\ $K_{i+1}^{(2)}$\end{tabular}
            \arrow[hookrightarrow, r] &
            \begin{tabular}{@{}c@{}c@{}} $K_{i+2}^{(1)}$ \\ \verteq \\ $K_{i+2}^{(2)}$\end{tabular}
            \arrow[hookrightarrow, r] & \cdots \\
            & & & K_{i}^{(2)} \arrow[hookrightarrow, ru] & & &
        \end{tikzcd}
        }}
    \end{equation}
\end{lemma}

\begin{proof}
    The proof follows from the fact that all maps are inclusion maps.
\end{proof}

Applying homology to \eqnref{bigcommutedgmcplx}, we obtain the following corollary.

\begin{corollary}[Commutative Diagrams on Homology]\label{cor:commutehomology}
    The following diagram~commutes:
    \begin{equation}\label{eqn:bigcommutedgmhom}
        {\tiny{
        \begin{tikzcd}
            & & & H(K_i^{(1)})  \arrow[rd] & & &\\
            \cdots \arrow[r] &
            \begin{tabular}{@{}c@{}c@{}}
                $H(K_{i-2}^{(1)})$  \\ \vcong \\ $H(K_{i-2}^{(2)})$
            \end{tabular}
            \arrow[r] &
            \begin{tabular}{@{}c@{}c@{}}
                $H(K_{i-1}^{(1)})$ \\ \vcong \\ $H(K_{i-1}^{(2)})$
            \end{tabular}
            \arrow[ru] \arrow[rd] &  &
            \begin{tabular}{@{}c@{}c@{}}
                $H(K_{i+1}^{(1)})$ \\ \vcong \\ $H(K_{i+1}^{(2)})$
            \end{tabular}
            \arrow[r] &
            \begin{tabular}{@{}c@{}c@{}}
                $H(K_{i+2}^{(1)})$ \\ \vcong \\ $H(K_{i+2}^{(2)})$
            \end{tabular}
            \arrow[r] & \cdots \\
            & & & H(K_{i}^{(2)}) \arrow[ru] & & &
        \end{tikzcd}
    }}
    \end{equation}
\end{corollary}


Suppose that $f_1$ and $f_2$ are $f$-compatible index filtrations that differ
by a single transposition. We show that it does not matter if we compute $f$
using $f_1$ or $f_2$. Following~\cite[p.~122]{vines}, there are four cases to
consider, namely, whether births or deaths occur at index $i$ and $i+1$.  We
first consider the case that $i$ and $i+1$ both correspond to birth events; the
other cases are similar.

\begin{lemma}[Transpose Births]\label{lem:transpose-births}
    Let $\simComp$ be a simplicial complex and $f \colon K \to \R$ be a filter
    function.  Let $f_1$ and $f_2$ be two $f$-compatible
    index filters that differ by exactly one transposition at index
    $i$ and $i+1$.
    If both~$i$ and~$i+1$ are birth coordinates in~$\astdgm_{f_1}$,
    then
    $$\left\{ \left(f\left(f_1^{-1}(i)\right), f\left(f_1^{-1}(j)\right) \right)
      \right\}_{(i,j) \in \astdgm_{f_t}}
      =
      \left\{ \big(f(f_2^{-1}(i)), f(f_2^{-1}(j)) \big)
      \right\}_{(i,j) \in \astdgm_{f_1}},$$
with $f(\emptyset):=\infty$ and where $=$ denotes equality as a multi-set.
\end{lemma}

\begin{proof}
    For $j \in \{1,2,\ldots, n\}$, let $\sigma_j:= f_1^{-1}(j)$; in other
    words,~$\sigma_1,\sigma_2, \ldots, \sigma_n$ is the ordering of the simplices
    such that $f_1(\sigma_j)=j$.
    Because the transposition is at indices~$i$ and~$i+1$,
    we have~$f_1(\sigma_i)=f_{2}(\sigma_{i+1})=i$
    and~$f_{2}(\sigma_{i+1})=f_1(\sigma_{i+1})=i+1$. For
    all other~$\sigma' \in K$,
    we have~$f_1(\sigma')=f_{2}(\sigma')$.

    Because~$i$ is a birth coordinate in~$\astdgm_{f_1}$, and because only a
    single simplex is added at index $i$, we know~$\rk(H(K_{i}^{(1)})) =
    \rk(H(K_{i-1}^{(1)})+1$ (where $\rk(H(-))$ denotes the sum of homology
    ranks over all dimensions).  Likewise, $\rk(H(K_{i+1}^{(1)}) =
    \rk(H(K_{i}^{(1)})+1$. See, e.g.,~\cite[Chapter V.4, pp.
    120--121]{edelsbrunner2010computational}. Thus, for all $k \in \N$, and
    using \corref{commutehomology}, we have
    \begin{equation}
        \rk(H_k(K_{i+1}^{(2)}) =\rk(H_k(K_{i+1}^{(1)})) = \rk(H(K_{i-1}^{(1)})) +2
        = \rk(H_k(K_{i-1}^{(2)}))+2.
    \end{equation}
    Then, because the rank of homology groups only increases
    and because~$K_{i-1}^{(2)}$ and $K_{i+1}^{(2)}$
    differ by two simplices, we know that both~$i$ and $i+1$ are
    also birth coordinates in~$\astdgm_{f_2}$.

    For $k \in \N$ and $s,t \in \Rbar$, let $\betti{1}{k}{s}{t} =
    \beta_{k,s,t}(f_{1})$ and $\betti{2}{k}{s}{t} = \beta_{k,s,t}(f_{2})$.
    By using~$f_1$ for the filter in \defref{apd} (and rephrasing \eqnref{multiplicity}
    specifically for index filtrations) the multiplicity of
    $(s',t')$ in~$\dgm{k}{f}$~is:
	\begin{align}
		\mu_k^{(1)}(s',t')
            &= \sum_{\substack{(s,t) \in \Rbar^2 \\ s < t \\ f(f_1^{-1}(s))=s'\\ f(f_1^{-1}(t))=t'}}
            \mu_k^f(s,t) \\
            &= \sum_{\substack{(s,t) \in \Rbar^2 \\ s< t\\ f(f_1^{-1}(s))=s'\\ f(f_1^{-1}(t))=t'}}
        \left( \betti{1}{k}{s}{t-1}
            - \betti{1}{k}{s}{t}
            - \betti{1}{k}{s-1}{t-1}
            + \betti{1}{k}{s-1}{t}
        \right).
		\label{eqn:multiplicitysum}
	\end{align}
    Similarly, let $\mu_k^{(2)} \colon \Rbar^2 \to \Z$ map $(s',t')$ to the
    multiplicity of $(s',t')$ in $\dgm{k}{f}$, as defined using $f_2$ as the filter in
    \defref{apd}.
    Now, we must show that
    $\mu_k^{(1)}=\mu_k^{(2)}$ for all $k$.

    We investigate the values of~$\betti{1}{k}{s}{t}$ for different values of
    $s$ and~$t$.  Because we are working with index
    filters,~$\betti{1}{k}{s}{t} = \betti{1}{k}{\floor{s}}{\floor{t}}$ and it is
    sufficient to only consider integer values of $s,t$.
    First, suppose $s \neq i$. Then, we also have~$t \neq i$. Because
    Diagram~\ref{eqn:bigcommutedgmhom} commutes,
    we know:
    \begin{align*}
        \betti{1}{k}{s}{t}
        &= \rk \left(  H_k( K_s^{(1)}) \to H_k( K_t^{(1)}) \right)\\
        &= \rk \left(  H_k( K_s^{(2)}) \to H_k( K_t^{(2)}) \right)\\
        &= \betti{2}{k}{s}{t}.
    \end{align*}

    Next, consider the case $s=i$.  Because~$s=i$ is a birth parameter, there
    exists~$d_i \in \Rbar$ such that $(i,d_i) \in \astdgm_{f_1}$. Similarly,
    let $d_{i+1} \in \Rbar$ such that $(i+1,d_{i+1}) \in \astdgm_{f_1}$.  Either
    let~$t \geq \max\{d_i, d_{i+1}\}$, or $t < \min(d_i,d_{i+1})$, or suppose $k$
    does not equal the dimension of death at $\max\{d_i, d_{i+1}\}$. Then,
    loosely speaking, $K_i^{(1)} \hookrightarrow K_t^{(1)}$
    and $K_i^{(2)} \hookrightarrow K_t^{(2)}$
    either do not include
    the deaths at $d_i$ and $d_{i+1}$, or they include
    both deaths. Either way, we have,
    \begin{align*}
        \betti{1}{k}{i}{t}
        &= \rk \left(  H_k( K_i^{(1)}) \to H_k( K_t^{(1)}) \right)
            && \text{by definition} \\
        &= \rk \left(  H_k( K_i^{(2)}) \to H_k( K_t^{(2)}) \right)\\
        &= \betti{2}{k}{i}{t}.
    \end{align*}
    Now, consider the case that $k$
    equals the dimension of death at $\max\{d_i, d_{i+1}\}$,
    and a value $t \in \N$ such that $t \in [\min\{d_i,d_{i+1}\},
    \max\{d_i, d_{i+1}\})$.
    First, suppose $d_i < d_{i+1}$, so that, loosely speaking, $K_i^{(1)}
    \hookrightarrow K_t^{(1)}$ includes the death at $d_i$ (but not at
    $d_{i+1}$), but~$K_i^{(2)} \hookrightarrow K_t^{(2)}$ does not include the
    death at $d_{i}$ (nor at $d_{i+1}$).
    Then,
    \begin{align*}
        \betti{1}{k}{i}{t}
        &= \rk \left(  H_k( K_i^{(1)}) \to H_k( K_t^{(1)}) \right)\\
        &= \rk \left(  H_k( K_i^{(2)}) \to H_k( K_t^{(2)}) \right) + 1 \\
        &= \betti{2}{k}{i}{t} + 1
    \end{align*}
	If $k$ equals the dimension of death at $\max\{d_i, d_{i+1}\}$, and $t
	\in \N$ such that~$t \in [\min\{d_i,d_{i+1}\}, \max\{d_i, d_{i+1}\})$,
	but instead, we have $d_i > d_{i+1}$, a symmetric argument~shows
    \begin{align*}
        \betti{1}{k}{i}{t}
        &= \betti{2}{k}{i}{t} - 1.
    \end{align*}

    Now that we have determined all values of $\betti{1}{k}{s}{t}$ and
    $\betti{2}{k}{s}{t}$, we evaluate the sum in \eqnref{multiplicitysum}.
    Again, we proceed by cases. In the first case, if
    $k$ does not equal the dimension of death at $\max\{d_i,
    d_{i+1}\}$, then for
    all~$(s', t')\in \Rbar^2$, we have~$\mu_k^{(1)}(s', t') =
    \mu_k^{(2)}(s',t')$.
    Then, suppose instead that $k$ does equal the dimension of death at $\max\{d_i,
    d_{i+1}\}$. Note that because $f_1$ and $f_2$ are compatible with $f$, we
    have~$f(f^{-1}_1)=f(f^{-1}_2)$. Furthermore, the sets $S = \{s \text{
	    s.t.  }f(f_1^{-1}(s)) = f(f_2^{-1}(s)) = s'\}$ and $T = \{t
    \text{ s.t.  }f(f_1^{-1}(t)) = f(f_2^{-1}(t)) = t'\}$ are each sets of
    consecutive integers (with the exception of the possible \mbox{set $T =
    \{\infty\}$}). We label elements
    $S = \{s_1, s_2 \ldots, s_r\}$, and write $s_0 = s_1-1$. By a
    telescoping argument, for any fixed $t \in T$, we find
    \begin{align*}
	    \sum_{s \in S} &\left( \betti{1}{k}{s}{t-1} -
	    \betti{1}{k}{s}{t} - \betti{1}{k}{s-1}{t-1} +
	    \betti{1}{k}{s-1}{t} \right) \\
	    &= \left( \betti{1}{k}{s_1}{t-1} -
	    \betti{1}{k}{s_1}{t} - \betti{1}{k}{s_0}{t-1} +
	    \betti{1}{k}{s_0}{t} \right)\\
	    &+ \left( \betti{1}{k}{s_2}{t-1} -
	    \betti{1}{k}{s_2}{t} - \betti{1}{k}{s_1}{t-1} +
	    \betti{1}{k}{s_1}{t} \right)\\
	    & \qquad \vdots \\
	    &+ \left( \betti{1}{k}{s_{r-1}}{t-1} -
	    \betti{1}{k}{s_{r-1}}{t} - \betti{1}{k}{s_{r-2}}{t-1} +
	    \betti{1}{k}{s_{r-2}}{t} \right)\\
	    &+ \left( \betti{1}{k}{s_r}{t-1} -
	    \betti{1}{k}{s_r}{t} - \betti{1}{k}{s_{r-1}}{t-1} +
	    \betti{1}{k}{s_{r-1}}{t} \right)\\
	    &= - \betti{1}{k}{s_0}{t-1} +
	    \betti{1}{k}{s_0}{t} + \betti{1}{k}{s_r}{t-1} -
	    \betti{1}{k}{s_r}{t},
    \end{align*}
    and similarly for $\betti{2}{k}{s}{t}$ when considering $f_2$.
    We will show that these sums for $f_1$ and~$f_2$ agree
    for any choice of $t$.

    Suppose for the fixed $t$ chosen above, $t \geq \max\{d_i, d_{i+1}\}$
    or $t < \min\{d_i,d_{i+1}\}$ and~$t-1 \geq \max\{d_i, d_{i+1}\}$ or $t -1
    < \min\{d_i,d_{i+1}\}$. For these values of $t$, we know
\begin{align*}
	-\betti{1}{k}{s_0}{t-1} +
	    \betti{1}{k}{s_0}{t} + \betti{1}{k}{s_r}{t-1} -
	    \betti{1}{k}{s_r}{t} = -\betti{2}{k}{s_0}{t-1} +
	    \betti{2}{k}{s_0}{t} + \betti{2}{k}{s_r}{t-1} -
	    \betti{2}{k}{s_r}{t}.
\end{align*}

    Suppose for the fixed $t$ chosen above, both $t, t-1 \in [\min\{d_i,
    d_{i+1}\}, \max\{d_i, d_{i+1}\})$ and $d_i < d_{i+1}$. Then, we know
\begin{align*}
	-\betti{1}{k}{s_0}{t-1} +
	    \betti{1}{k}{s_0}{t} + &\betti{1}{k}{s_r}{t-1} -
	    \betti{1}{k}{s_r}{t} \\
	    &= -(\betti{2}{k}{s_0}{t-1}+1) +
	    (\betti{2}{k}{s_0}{t}+1) + (\betti{2}{k}{s_r}{t-1}+1) -
	    (\betti{2}{k}{s_r}{t}+1) \\
	    &= -\betti{2}{k}{s_0}{t-1} +
	    \betti{2}{k}{s_0}{t} + \betti{2}{k}{s_r}{t-1} -
	    \betti{2}{k}{s_r}{t}.
\end{align*}
For the same $t$ and $t-1$, if $d_i > d_{i+1}$, a nearly identical argument
holds.

If $t  = \min\{d_i,
    d_{i+1}\}$ and $d_i < d_{i+1}$, then we have
\begin{align*}
	-\betti{1}{k}{s_0}{t-1} +
	    \betti{1}{k}{s_0}{t} + &\betti{1}{k}{s_r}{t-1} -
	    \betti{1}{k}{s_r}{t} \\
	    &= -\betti{2}{k}{s_0}{t-1} +
	    (\betti{2}{k}{s_0}{t}+1) + \betti{2}{k}{s_r}{t-1} -
	    (\betti{2}{k}{s_r}{t}+1) \\
	    &= -\betti{2}{k}{s_0}{t-1} +
	    \betti{2}{k}{s_0}{t} + \betti{2}{k}{s_r}{t-1} -
	    \betti{2}{k}{s_r}{t}.
\end{align*}
Again, for the same $t$ and $t-1$, if $d_i > d_{i+1}$, a nearly identical argument
holds.

    Note that because we cannot have $t-1 \in [\min\{d_i, d_{i+1}\}, \max\{d_i,
    d_{i+1}\})$ without~$t$ also being in this set, we have evaluated all
    possible cases of $t$ and $t-1$.  In each case, the term of the sum
    corresponding to \eqnref{multiplicitysum} agreed between $\mu^{(1)}(s',t')$
    and~$\mu^{(2)}(s',t')$. This, combined with our previous arguments, means we
    have shown the equality $\mu^{(1)}(s',t') = \mu^{(2)}(s',t')$, as desired.
\end{proof}

We have, so far, only considered the case that $i$ and $i+1$ correspond to
birth events. The remaining three cases (i.e., that they correspond to death
events, a birth and death event, or a death and birth event) are similar.
All cases taken together give us the following corollary.

\begin{corollary}[Single Transposition]
	\label{cor:sing-transpositions}
    Let $\simComp$ be a simplicial complex, $m \in \N$, and~$f \colon K \to \R$ a filter
    function.  Let $f_1$ and $f_2$ be two $f$-compatible
    index functions that differ by exactly one transposition.
    Then:
    $$\left\{ \left(f\left(f_1^{-1}(i)\right), f\left(f_1^{-1}(j)\right) \right)
      \right\}_{(i,j) \in \astdgm_{f_t}}
      =
      \left\{ \big(f(f_2^{-1}(i)), f(f_2^{-1}(j)) \big)
      \right\}_{(i,j) \in \astdgm_{f_1}},$$
with $f(\emptyset):=\infty$ and where $=$ denotes equality as a multi-set.
\end{corollary}

Next, we move to transposition sequences of arbitrary finite length, and show
that all index filtrations in a transposition sequence produce the same
verbose persistence~diagram.

\begin{lemma}[Transposition]\label{lem:transpositions}
    Let $\simComp$ be a simplicial complex, $m \in \N$, and $f \colon K \to \R$ a filter
    function.  Let $\{f_1,f_2,\ldots, f_m\}$ be an $f$-compatible
    transposition sequence.
    Then, for all $t \in \{ 1,2, \ldots, m\}$,
    $$\left\{ \left(f\left(f_t^{-1}(i)\right), f\left(f_t^{-1}(j)\right) \right)
      \right\}_{(i,j) \in \stdgm{k}{f_t}}
      =
      \left\{ \big(f(f_1^{-1}(i)), f(f_1^{-1}(j)) \big)
      \right\}_{(i,j) \in \stdgm{k}{f_1}},$$
with $f(\emptyset):=\infty$ and where $=$ denotes equality as a multi-set.
\end{lemma}

\begin{proof}
    For $t \in \{1,2, \ldots, m\}$, we write
    \begin{equation}\label{eqn:multiset}
        {A}_t := \left\{ \left(
                f \left({f}_t^{-1}(i)\right),
                f\left({f}_t^{-1}(j)\right) \right)
        \right\}_{(i,j) \in \astdgm_{f_t}}.
    \end{equation}
    We prove the general claim by induction on $m$, the length of $f$-compatible
    transition sequence.
    The base case is straightforward: Let~$f_1$ be an $f$-compatible index
    filtration. We immediately have~$A_1=A_1$.
    %

    For the inductive assumption, let~$m \geq 1$, and suppose for any
    $f$-compatible transposition sequence,
    $\{f_1,f_2, \ldots, f_m \}$, we have $A_{m'}=A_1$ for
    all~$m' \in \{1,2, \ldots, m \}$.
    %

    Now, consider $\{f_1,f_2,\ldots, f_{m+1}\}$, an $f$-compatible
    transposition sequence of length~$m+1$.
    By the inductive assumption, we know that $A_2 = A_3 = \ldots = A_{m+1}$, so it
    suffices to show that $A_1 = A_{2}$.

    Recall that $f_1$ and $f_{2}$ differ by one transposition, so let $\sigma$
	and~$\tau$ be the simplices transposed and suppose the transposition
	occurs at indexes $i$ and $i+1$. By
	\corref{sing-transpositions}, we have $A_1 = A_2$. Thus, we have shown the claim.
\end{proof}

Finally, we show all index filtrations compatible with a common
underlying filtration all produce the same VPD, i.e., VPDs are well-defined.

\begin{lemma}[VPD is Well-Defined]\label{lem:well-defined}
    The verbose persistence diagram
    is well-defined.
\end{lemma}

\begin{proof}
    Let $\simComp$ be a simplicial complex, and
    let~$f \colon \simComp \to \R$ be
    a filter function.
    Let~$f', f'' \colon \simComp \to \R$
    be two index filters compatible with $f$.
    To show that the \apd is well-defined, we
    show that $f'$ and~$f''$ produce the same \apd.

    Because~$f'$ and $f''$ are compatible with $f$ and produce a total order on
    the simplices of $\simComp$, there exists a sequence of
    filters $\{ f_1,f_2, \ldots, f_m \}$ such that $f_1=f'$, $f_m=f''$, and
    consecutive filter functions agree except for
    one $f$-compatible transposition.
    As in \eqnref{apd-pts} of \defref{apd}, define the following~multisets for
    $t \in \{1,2,\ldots, m\}$:
    \begin{equation}\label{eqn:multiset}
        {A}_t := \left\{ \left(
                f \left({f}_t^{-1}(i)\right),
                f\left({f}_t^{-1}(j)\right) \right)
        \right\}_{(i,j) \in \astdgm_{f_t}}.
    \end{equation}
    By \lemref{transpositions}, we know that $A_t=A_1$ for all $t$. In
    particular, $A_m=A_1$.  Thus, both
    $f'$ and $f''$ produce the same verbose persistence diagram, which means
    that VPDs are well-defined.
\end{proof}

To show that \abcs and \aeccs are well-defined, we first observe the following
relationship between descriptors, arising from the connection between homology,
Betti number, and Euler characteristic.

\begin{observation}\label{obs:reduction}
    Given the verbose persistence diagram $\adgm_{f}$ for some simplicial
    complex $\simComp$ and filter function $f\colon \simComp \to \R$, we know
    the Betti numbers at every step of a corresponding index
    filtration and can thus
    construct the verbose Betti function $\augbc_{f}$. Furthermore, we can
    construct the verbose Euler characteristic function $\acrv_f$ by taking the
    alternating sum of these Betti numbers.
\end{observation}

Then, well-definedness of both the \abc and \aecc are straightforward corollaries of
\lemref{well-defined}; because two distinct index filters compatible with the same
filter function produce the same \apd, they produce the same \abc and the same \aecc.

\begin{corollary}[VBF and VECF are
    Well-Defined]\label{cor:well-defined-abc}
    Let $\simComp$ be a simplicial complex and let~$f \colon \simComp \to \R$ be
    a filter function.  Then, the verbose Betti function
	$\augbc_{f}$ and the verbose Euler characteristic function
    $\acrv_f$ are well-defined.
\end{corollary}

%% file: body/append-count.tex
In this appendix, we provide the proof of \lemref{count}. We first recall
the statement of this well-known result.

\simpcount*

\begin{proof}
    Let $f' \colon \simComp \to \N$ be an index filter function compatible with
    $f$, and let $\{F'_i \}_{i \in \N}$ be
    the corresponding index filtration, where $F'_i := f'^{-1}(-\infty,i]$.
    Let $\sigma_1,\sigma_2, \ldots, \sigma_n$ be the ordering of simplices
    in $K$ such that $f'(\sigma_i)=i$.
    Consider the sets
    \begin{equation*}
        C_B := \left\{ (i,j) \in \stdgm{k}{f'}
            \text{ s.t. } f(\sigma_i) = c \right\}
        \quad
        C_D := \left\{ (i,j) \in \stdgm{k-1}{f'}
            \text{ s.t. } f(\sigma_j)= c \right\}
    \end{equation*}
    and let $C= C_B \cup C_D$.

    We start by defining a bijection $\phi \colon f^{-1}(c) \to C$.
    Let $\sigma_i \in \simComp_k$ such that $f(\sigma_i)=c$.
    Because adding a single $k$-simplex either increases $\beta_k$ by one or
    decreases $\beta_{k-1}$ (see,
    e.g.,~\cite[pp. 120--121]{edelsbrunner2010computational}),
    either
    $\beta_k(F'_{i})=\beta_k(F'_{i-1})+1$ or
    $\beta_{k-1}(F'_{i})=\beta_{k-1}(F'_{i-1})-1$, but not
    both.
    \begin{enumerate}[\hspace{18pt}(P1)]
        \item[Case 1:]  ($\beta_k$ increases). There exists a unique birth at
            index~$i$ in
            $\stdgm{k}{f'}$.  Thus, let $j \in \Rbar$ be such that~$(i,j) \in
\stdgm{k}{f'}$.
            Then, $(f'(\sigma_i),f'(\sigma_j))=(i,j) \in C_B$.
            Define $\phi(\sigma_i)=(i,j)$.
        \item[Case 2:]  ($\beta_{k-1}$ decreases). There is a unique death at
            index~$i$
            in~$\stdgm{k}{f'}$. Thus, let $j \in \Rbar$ be such that~$(j,i)
\in \stdgm{k}{f'}$.
            Then, $(f'(\sigma_j),f'(\sigma_i))=(j,i) \in C_D$.
            Define $\phi(\sigma_i)=((f'(\sigma_j),f'(\sigma_i))=(j,i)$.
    \end{enumerate}
    In other words, each $\sigma_i$ is mapped to the persistence pair containing
    $f'(\sigma_i)$. If~$\phi(\sigma)=\phi(\tau)=(i,j)$, then, by construction,
    both $\sigma$ and $\tau$ are the same type of event. WLOG, assume that
    they are birth events.  Then,~$f'(\sigma)=f'(\tau)$, which, by injectivity
    of $f'$ means that $\sigma=\tau$.
    Thus, we have shown that $\phi$ is an~injection.

    We show that $\phi$ is a surjection by contradiction.
    Suppose there \mbox{exists $(i,j) \in C$} such
    that there does not exist $\sigma \in f^{-1}(c)$ with $\phi(\sigma)=(i,j)$.
    Because~$(i,j)\in C$, then either~$f(\sigma_i)=c$ or $f(\sigma_j)=c$ (or
    both). WLOG, suppose $f(\sigma_i)=c$.  Then, we \mbox{find~$\sigma_i \in
    f^{-1}(c)$} and, by construction, $\phi(\sigma_i)=(i,j)\in C_B$.

    By \eqnref{apd-pts} and because~$\adgm_{f}$ is well-defined (see
    \lemref{well-defined}), we know that the map $\psi \colon \stdgm{k}{f'} \to
\dgm{k}{f}$ defined by
    $\psi(i,j)=(f(\sigma_i),f(\sigma_j))$ is a bijection.  Because~$C \subset
    \stdgm{k}{f'}$ and both functions are bijections, the composition $\psi
    \circ \phi$ is the one-to-one correspondence that we sought.
\end{proof}
Thus, if some step in a filtration includes a single additional simplex, there
must be a single additional point in the diagram.
This result is noted elsewhere, e.g.,~\cite[Chapter
V.4]{edelsbrunner2010computational}, and we state it as a corollary to
\lemref{count}.

\begin{corollary}[Add One Simplex]\label{cor:add-one}
        Let $K_i$ and $K_{i+1}$ be two simplicial complexes such that $K_{i+1}
        \setminus K_i$ is a single $k$-simplex.  Let $\beta_k^{(i)}$ and
        $\beta_k^{(i+1)}$ denote the $k$th Betti numbers of $K_i$ and
        $K_{i+1}$, respectively.  Then,
        either $\beta_k^{(i+1)} = \beta_k^{(i)} + 1$
        or~$\beta_{k-1}^{(i+1)} = \beta_{k-1}^{(i)} - 1$. No other Betti numbers
        change values.
\end{corollary}

%% file: body/append-omitted.tex
In this appendix, we provide algorithms and proofs omitted from the main text.

\subsection{Proofs from \secref{mainproof-vect}}
\label{append:omitted-vect}
We begin with the proof of \thmref{evenindegValue}, which we restate below.

\evenindegVal*

\begin{proof}
    We begin by considering even-indegree.
 The proof is similar to the proof of \thmref{indegValue}.
    Because $S$ is assumed to be simplex-isolating
    for~$K$, if $S$ is not simplex-isolating for $\sigma$, then $\sigma$ is not a
    simplex of~$K$ and thus must have even-indegree 0, which is the value
    returned on \alglnref{eveninkdeg:notasimplex}.

    Suppose then that $S$ is simplex-isolating for $\sigma$.
    We prove the claim
    inductively on~$j= \dim(\sigma)$ and let $h_s: K \to \R$ denote the filter
    function for direction $s$.
    For the base case ($j=0$), consider $\sigma = [v]$, a single zero-simplex.
    Because $\dir$ is $K_0$-perpendicular to $\sigma$, the even-indegree of
    $v$ is exactly equal to the number of~even-simplices that have
    height~$\dir \cdot v$, which can immediately be read off of $\acrv_{h_s}$. In
    \algref{evenindeg}, notice that because the vertex $v$ does not have any
    proper faces, we do not
    enter the loop that starts on \alglnref{inkdeg:for}.  Thus, the return value
    is exactly the number of even simplices with height~$\dir \cdot v$.

For the inductive step, let $j \geq 0$.  We assume that \algref{evenindeg}
    returns the~even-indegree of $\tau$ in direction~$\dir$, for all $\tau \in
    \simComp_j$.
	Suppose that~$\dim(\sigma) = j+1$.
    Now, we compute the~even-indegree of~$\sigma$ in direction~$\dir$.
    Let $F_\sigma$ denote the set of~even-dimensional simplices with
    height~$h_\dir(\sigma)$ in direction~$\dir$, whose cardinality is reported
    in $\acrv_{h_\dir}$. Let $\sigma' \in F_{\sigma}$
    and let $\tau \prec \sigma$. Suppose that $\dir_\tau$ is a~$(K_0, {\sigma},
    {\sigma \setminus \tau}, \dir)$-perturbation.  By \lemref{aecc_uniqueiso},
    we know
    $\sigma'$ contributes to the even-indegree of $\tau$ in
    direction $\dir_\tau$ if and only if~$\tau = \sigma \cap \sigma'$.

    Recall that we have a nearly identical version of \corref{subtract} for the
    setting of even/odd-indegree, which gives
    a simplified equation
    for the even-indegree of $\sigma$:
    \begin{equation}\label{eqn:computeevenindegree}
        \delta - \sum_{\tau \prec \sigma} \delta_\tau,
    \end{equation}
    where $\delta$ is the number of even-simplices at the height of $\sigma$
    and $\delta_\tau$ is the even-indegree of $\tau$ in the corresponding
    tilted directions.  In \algref{evenindeg}, \varName{allEven} is equal to
    $\delta$, and the values $\delta_{\tau}$ are computed in
    \alglnref{inkdeg:doublecount}.  Thus, the return value of
    \algref{evenindeg} matches the value of \eqnref{computeevenindegree}.

    The proof for odd-indegree is nearly identical.
\end{proof}

Next, we restate and prove \thmref{aeccmain}.

\aeccmain*

\begin{proof}
    The initial steps of the proof follow the reasoning given in
    \thmref{full-reconstruction}; we reconstruct $K_0$, and when we
    consider some $V$ with $\lvert V \rvert -1 = k$, we have already reconstructed
    $K_{k-1}$.  First, suppose $k$ is even.  Because $\apht{K}{S}$ is
    simplex-isolating, if~$V$ defines a simplex of $K$, the set $S$
    contains a direction $s$ that is $K_0$-perpendicular to~$\aff{V}$.
    Thus, if there is no such direction, $V$ does not define a
    simplex of~$K$.

    If there is such a direction, by \thmref{evenindegValue},
    $\evenindegAlgLong{\sigma}{\dir}{\aecct{K}{S}}{T = \{ \}}$
    (\algref{evenindeg}) returns the number of even-simplices at the height
    of $\aff{V}$ that contain $\sigma$ as a face; if this value is 0, we
    know $V$ is not a simplex of $K$. If this value is~$1$ or greater,
    we know $V$ does define a simplex of $K$, and we add $V$ to $K_k$.
    A nearly identical argument holds for the case that $k$ is odd. Because
    we iterate
    over all subsets of~$K_0$, the algorithm eventually finds all
    simplices.

    Finally, because we have shown $\aeccreconAlgo$ reconstructs~$K$, we
    know~$\aecct{K}{S}$ is a faithful discretization of $\aecct{K}{\S^{d-1}}$.
\end{proof}

    \subsection{Algorithms and Proofs from \secref{aux}}
    \label{append:omitted-aux}
In this appendix, we provide algorithms and proofs that were omitted from~\secref{aux}.

\paragraph{Computing a Perpendicular Direction}
\input{body/algorithms/perp}
\begin{restatable}[Computing a Perpendicular Direction]{lemma}{gs}\label{lem:gs}
    \algref{perp} is correct and runs in $\perpTimeFull$ time.
\end{restatable}

\begin{proof}
    Denote the points of $P$ by $p_0, p_1, \ldots, p_{k}$.  In
    \alglnref{perp:matrix}, we define $A$ as
    the $d \times k$ matrix consisting of $p_j - p_0$ in column $j$. By Theorem
    $7.1$ in \cite{trefethen1997numerical}, $A$ has a reduced QR-factorization
    in which the columns of $Q$ are an orthonormal basis for the space spanned
    by $A$, or $\aff{P}$. We
    find $Q$ using $k$ iterations of Gram--Schmidt on \alglnref{perp:qn}.

    Then, in the loop at \alglnrefRange{perp:whilebegin}{perp:whileend}, we
    iteratively test the result of performing Gram--Schmidt on the columns of $Q$
    and a basis vector. The result is zero if the basis vector is contained
    in $\aff{P}$. As soon as the result is nonzero, (which has to happen
    because~$\aff{P}$ has positive codimension, so can't be spanned by all $e_i$s), we
    have found a vector orthogonal to $\aff{P}$, and the value returned on
    \alglnref{perp:dir} is as desired.

    Building the matrix $A$ and normalizing $\dir$ takes constant time. Finding
    the reduced~QR-factorization takes $k$ iterations of Gram--Schmidt,
    (see~\cite[Theorem 8.1]{trefethen1997numerical}) and finding the final
    vector that is
    orthogonal to the space spanned by $P$ requires at most~$d$ iterations of
    Gram--Schmidt. Because each iteration of Gram--Schmidt takes~$O(dk)$ time, we see
    that the total runtime is $\perpTimeFull$.
\end{proof}

\paragraph{Plane Filling}
\input{body/algorithms/planefill}

\begin{lemma}[Plane Filling] \label{lem:planefill}
    \algref{planefill} is correct and runs in $\planefillTimeFull$~time.
\end{lemma}

\begin{proof}
    In \alglnrefTwo{planefill:matrix}{planefill:qn} of \algref{planefill}, we
    find an orthonormal basis for the space spanned by $\{p_1-p_0, p_2-p_0, \ldots,
    p_k-p_0,s\}$ in $\Theta(kd)$ time, following the same process we did in
    \alglnrefTwo{perp:matrix}{perp:qn} of \algref{perp}; see proof of
    \lemref{gs}.

    Then, by using
    Gram--Schmidt in \alglnrefRange{planefill:forbegin}{planefill:forend}, we
    find $d-(k+1)$ vectors spanning the space orthogonal to $\aff{A}$,
    taking a total of~$\Theta(d^3)$ time.
    While doing so, we also
    iteratively add points to~$P'$ in \alglnref{planefill:points}.
    Computing the full QR-decomposition
    and performing Gram--Schmidt dominates the algorithm; hence, the running
    time for \algref{planefill} is $\planefillTimeFull$.

    Finally, for correctness, we prove that the output of \algref{planefill} has
    the desired properties.  By construction, we have $\lvert P'\rvert = d -
    \lvert P \rvert$, i.e., $\dim(\aff{P' \cup P}) = d-1$.  To show each point
    in $P'$ is at the same height as $P$, it suffices to show that~$\dir \cdot
    p_0 = \dir \cdot p'$ for all $p' \in P'$. Because each $s_\ell$ considered in
    \alglnref{planefill:points} is in the basis
    for the left null space of $A$ and $s$ is a column of~$A$, all such vectors
    $s_\ell$ are orthogonal to~$\dir$.  Then, indeed, for all~$p' \in P'$,
    we have~$\dir \cdot p' = \dir \cdot (s_\ell + p_0) = \dir \cdot p_0$.
\end{proof}

%% file: body/algorithms/perp.tex
\begin{algorithm}[h!]
    \caption{$\perpendicular{P}$}\label{alg:perp}
\begin{algorithmic}[1]
    \REQUIRE $P = \{p_0, p_1, \ldots, p_{k}\} \subset \R^d$, a set of affinely independent points such that
        $k < d$.
    \ENSURE $\dir \in \sph^{d-1}$, a direction that is orthogonal to $\aff{P}$.
    \STATE $A \gets$  the $d \times k$ matrix containing $p_j - p_0$ in column $i$ for
        $1 \leq j \leq k$
        \label{algln:perp:matrix}
    \STATE $B \gets$ orthonormal basis vectors for the space spanned by $A$
         \label{algln:perp:qn}
    \STATE $i \gets 1$
    \STATE $\dir \gets 0$
    \WHILE{$\dir = 0$}\label{algln:perp:whilebegin}
        \STATE $\dir \gets$ the output of Gram--Schmidt on $e_i$ using $B$
        \STATE $i \gets i+1$
    \ENDWHILE \label{algln:perp:whileend}
    \RETURN $\dir$
    \label{algln:perp:dir}
\end{algorithmic}
\end{algorithm}

%% file: body/algorithms/planefill.tex
\begin{algorithm}[h!]
    \caption{$\planefill{P}{\dir}$}\label{alg:planefill}
\begin{algorithmic}[1]
    \REQUIRE $P = \{p_0, p_1, \ldots, p_{k}\} \subset \R^d$,
        a set of affinely independent points such that
        $k < d$ and
        $\dir \in \sph^{d-1}$, a direction that is orthogonal to $\aff{P}$.
    \ENSURE a set $P'\subset \R^d$ of $d-\lvert P \rvert$ points at the same height as $P$
    with respect to direction~$s$
    such that~$\dim(\aff{P' \cup P}) = d-1$.
    \STATE $P' \gets \emptyset$
    \STATE $A \gets$  the $d \times (k+1)$ matrix containing $p_j - p_0$ in column $j$
        for $1 \leq j \leq k$ and~$s$ in the $(k+1)$st column
		\label{algln:planefill:matrix}
    \STATE $B \gets$ orthonormal basis vectors for the space spanned by $A$
        \label{algln:planefill:qn}
    \STATE $\ell \gets 0$
    \FOR{$i$ from $1$ to $d$}\label{algln:planefill:forbegin}
    \WHILE{$\ell < d-(k+1)$}\label{algln:planefill:pointbegin}
        \STATE $\dir_\ell \gets$ the output of Gram--Schmidt on $e_i$ using $B$
            \label{algln:planefill:gs}
        \IF{$\dir_\ell \neq 0$}
	    \STATE Add $\dir_\ell + p_0$ to $P'$
            \label{algln:planefill:points}
            \STATE $\ell \gets \ell+1$
        \ENDIF
    \ENDWHILE \label{algln:planefill:pointend}
    \ENDFOR\label{algln:planefill:forend}
    \RETURN $P'$
\end{algorithmic}
\end{algorithm}

%% file: body/example.tex
Here, we walk through a small example of building a faithful discretization of
the VPHT, VBT, or VECT, following \conref{explicit}.
Suppose we are given the simplicial complex~$K$ in~$\R^3$, shown in
\figref{example}.  Specifically,~$K$
is: five vertices, $v_0 = (1,0,0)$, $v_1 = (0,1,0)$,
$v_2 = (0,0,1)$,~$v_3=(0,2,0)$, and~$v_4=(1,0,2)$; four edges,~$[v_0,v_1]$,
$[v_1,v_2]$, $[v_0,v_2]$, and $[v_0,v_3]$; and a two-simplex $[v_0,v_1,v_2]$.
There are three maximal simplices,~$v_4$, $\sigma_1=[v_0,v_3]$, and $\sigma_2=[v_0,v_1,v_2]$.
\begin{figure}[h!]
    \centering
    \includegraphics[height=1.5in]{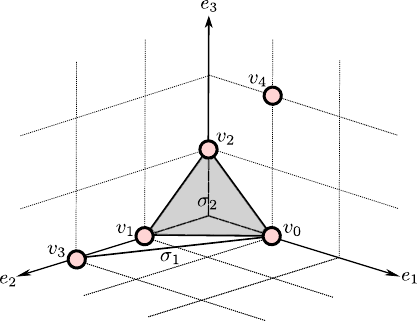}
    \caption{The simplicial complex $K$ used as an example in this section. We
    first find vertex-isolating directions: the basis vectors, along with a
    fourth direction that orders the filtration grid of $K_0$ with respect to
    the basis vectors uniquely. Next, for $\sigma_1$ and
    $\sigma_2$, we find simplex-isolating directions; directions that are
    $K_0$-perpendicular to $\sigma_1$ and $\sigma_2$, as well as directions that
    tilt these perpendicular directions so that proper faces of $\sigma_1$ (and
    $\sigma_2$) are ``popped'' above $\sigma_1$ (respectively, $\sigma_2$).
    }
    \label{fig:example}
\end{figure}

\subsection{Tilt Examples}\label{sec:tiltalgo}
Because many algorithms used in \conref{explicit} make calls to \algref{tilt}
($\algoName{Tilt}$), we walk through two explicit examples in this section,
using inputs that are relevant for later subsections.

We first describe the details of the call $\tiltraw{s}{s'}{P}$ where
$s=(1,0,0)$, $s'=(0,1,0)$ and $P = \{0,1,2\}^3 \subset \R^3$.  We begin by
finding the heights of points in $P$ with respect to the direction $s$ that are
closest but distinct (\alglnref{tilt:closest} of \algref{tilt}). There are only
two distinct heights of points of $P$ with respect to the direction $s$, namely
$0$ and $1$, so we have $h_1 = 0$ and $h_2 = 1$. Next, on
\alglnref{tilt:extremal}, we find heights of points of $P$ in the direction
$s'$ that are extremal, which are $h_1' = 0$ and $h_2' = 2$. Then, on
\alglnref{tilt:tstar}, we compute the solution to~$(1-t) h_1 + t h_2' =
(1-t)h_2 + t h_1'$, which is $t_* = 1/3$.
Finally, we return the tilted direction
\begin{align*}
    s_t = \left(1-\frac{1}{2}t_*\right)e_1+
\frac{1}{2}t_* e_2 = \left(\frac{5}{6}, \frac{1}{6}, 0\right) \approx (0.8333, 0.1667, 0).
\end{align*}

Next, we detail a more involved instance of $\algoName{Tilt}$, using
$s = (2\sqrt{5}/5, \sqrt{5}/5,
0)$, $s' =  (0.8, 0.4, 0.4472)$, and $P = K_0$.
First, on \alglnref{tilt:closest}, we identify two heights of vertices
with respect to $\dir$ that are closest. The direction~$\dir$ orders
vertices in an evenly spaced way~($v_0, v_3,$ and $v_4$ have
height~$2 \sqrt{5}/5$,~$v_1$ has height $\sqrt{5}/5$, and~$v_2$ has height~$0$).
Going through the sorted list, we first
encounter~$h_1 = 0$ and $h_2 = \sqrt{5}/5$. On \alglnref{tilt:extremal}, we find
the minimum and maximum heights with respect to $\dir'$. These are $h_1' = 0.4$
(the height of $v_1$ with respect to $\dir'$) and $h_2'
 \approx 1.6944$ (the height of
$v_4$ with respect to~$\dir'$). Finally, on
\alglnref{tilt:tstar} of \algref{tilt}, we find the solution to
$(1-t) h_1 + t h_2' = (1-t)h_2 + t h_1'$,
which is $t_* \approx 0.2568$.
This gives us the final output of $\tiltraw{\dir}{\dir'}{K_0}$,
\begin{align*}
\dir_t = \left(1-\frac{1}{2}t_*\right)\dir + \frac{1}{2}t_* \dir'
\approx (0.8823, 0.4412, 0.0574).
\end{align*}

\subsection{Vertex Isolating Directions}\label{sec:vertexisoexample}
First, we find the directions of \stepref{verts} of \conref{explicit}.
Specifically, we find directions that are vertex-isolating with
respect to~$K$.
By~\lemref{vertdirs}, the standard basis directions, $e_1, e_2, e_3$, as well as
a fourth direction as computed in $\pointiso{K_0}$ (\algref{pointiso}) are vertex-isolating with
respect to~$K$.  Thus, we walk through the computation of \algref{pointiso}.
First, on
\alglnref{pointiso:defineA}, we begin by defining $A$, the filtration grid of
$K_0$ with respect to
$\{e_1, e_2, \ldots, e_d\}$. In our example, $A$ is a set
of~$27$ gridpoints,~$\{0,1,2\}^3 \subset \R^3$.  In \alglnref{pointiso:dirinit}, we initialize $s =
(1,0,0)$, and enter the loop on
\alglnrefRange{pointiso:forBegin}{pointiso:forEnd}. The first iteration calls
$\tiltraw{(1,0,0)}{(0,1,0)}{A}$ (\algref{tilt}), tilting $e_1$ towards $e_2$.
As detailed in \secref{tiltalgo}, this gives us $s^2 = (1-\frac{1}{2}t_*)e_1+
\frac{1}{2}t_* e_2 \approx (0.8333, 0.1667, 0)$.
We are then ready for the next and final iteration of the loop in
\algref{pointiso}, finding the fourth direction in our set of vertex-isolating
directions,
\begin{equation}
    s = \tiltraw{s^2}{(0,0,1)}{A}
    \approx (0.8013, 0.1603, 0.0385).
\end{equation}

\subsection{Perpendicular Directions}\label{ssec:perpdirsexample}
Next, we compute directions as described in \substepref{higher}{perp} of
\conref{explicit} (equivalently, directions satisfying
\stmtref{prop:simplexiso-perp} of \defref{simplexiso}).
Then, in the language of \defref{simplexiso}, for $\sigma_i$ ($i = 1,2$) we
need a direction $\dir_{\sigma_i}$ that is~$K_0$-perpendicular to $\sigma_i$.
We begin with the two-simplex, $\sigma_2$, the more straightforward
computation.  First, we compute~$\perpendicular{\sigma_2}$ (\algref{perp}). In
\alglnref{perp:matrix} of \algref{perp}, we find a matrix $A$ with~$v_j - v_0$
in the $j$th column, which, for $\sigma_2$ is
\begin{equation}
    A = \begin{bmatrix}
        -1 & -1 \\
        1 & 0 \\
        0 & 1
    \end{bmatrix}.
\end{equation}
Next, in \alglnref{perp:qn}, we use a QR-decomposition to find the basis
vectors for the space spanned by $A$, which are, $b_1 =
(-\sqrt{2}/2,\sqrt{2}/2,0)$ and $b_2 = (-\sqrt{6}/6, -\sqrt{6}/6, \sqrt{6}/3)$.
Then, in \alglnrefRange{perp:whilebegin}{perp:whileend}, we perform
Gram-Schmidt on $b_1$ and $b_2$ with each basis vector until the result is
non-zero. The loop terminates after one iteration and we find
\begin{equation}
    \dir_{\sigma_2}
    = \perpendicular{\sigma_2}
    = (\sqrt{3}/3, \sqrt{3}/3, \sqrt{3}/3)
    \approx (0.5774, 0.5774, 0.5774).
\end{equation}
The
direction $\dir_{\sigma_2}$ is already $K_0$-perpendicular to $\sigma_2$;
we confirm this by completing the procedure suggested in
\substepref{higher}{perp} of \conref{explicit}.
Let $V$ be the set of vertices with the same height as $\sigma_2$ in
direction $\dir_{\sigma_2}$, meaning~$V = \verts{\sigma}$. Then, we
set~$W = V \setminus \sigma = \emptyset$, so~$\pop{K_0}{V}{W}{\dir_{\sigma_2}} =
\pop{K_0}{\sigma_2}{\emptyset}{\dir_{\sigma_2}} = \dir_{\sigma_2}$ (this is
highlighting the fact that, because the output of $\perpendicular{\sigma_2}$ is
already $K_0$-perpendicular to $\sigma_2$, we do not need to adapt it in any
way).

Next, we shift focus to $\sigma_1$.
To construct the direction described by \substepref{higher}{perp} of
\conref{explicit} for the edge $\sigma_1$, we again begin by finding
$\perpendicular{\sigma_1}$. This begins with the
column matrix containing $v_3 -
v_0$,
\begin{equation}
    A_1 = \begin{bmatrix}
         -1 \\
         2 \\
         0
    \end{bmatrix}.
\end{equation}
Then, we use QR-decomposition to find the basis for the space spanned by~$A_1$,
which consists of the single vector $b = (-\sqrt{5}/5, 2\sqrt{5}/5, 0)$.
Performing Gram-Schmidt on $b$ and $e_1$ yields $\dir =
\perpendicular{\sigma_1} = (2\sqrt{5}/5, \sqrt{5}/5, 0)$. However, note that
although $\dir$ is perpendicular to $\sigma_1$, it is not $K_0$-perpendicular,
because the vertex $v_4$ is at the same height as $\sigma_1$ with respect to
$\dir$. Thus, we correct~$\dir$ by following the procedure suggested in
\conref{explicit}, \substepref{higher}{perp}. First, let $V_1$ be the set of
all vertices with the same height as $\sigma_1$ in direction $\dir$, i.e.,
$V_1= \{v_0, v_3, v_4\}$. We also set $W_1 = V_1 \setminus \sigma_1 = \{v_4\}$.
Then
we use $\pop{K_0}{V_1}{W_1}{\dir}$ (\algref{pivot}) to find a
direction~$K_0$-perpendicular to $\sigma_1$. In \alglnref{tiltdown-planefill},
we find~$V' = \planefill{V_1}{\dir}$ but because~$\aff{V_1}$ has codimension
one with $\R^3$, we simply \mbox{have~$V' = \emptyset$}.
Next, in \alglnref{tiltdown-findperp}, we compute
\begin{align*}
    \dir' &= \perpendicular{{V' \cup (V_1 \setminus W_1) \cup \{W_1 -
    \dir\}}} \\
    &= \perpendicular{{v_0, v_3, v_4 - \dir}} \\
    &= \perpendicular{{v_0, v_3, (1-2\sqrt{5}/5, -\sqrt{5}/5, 2)}} \\
    & \approx (0.8, 0.4, 0.4467).
\end{align*}
Although for this example, $\dir'$ happens to be already
$K_0$ perpendicular to $\sigma_1$, this is not generally the
case. Thus, we continue on \alglnref{tiltdown-tilt} and compute
\begin{align*}
\dir_{\sigma_1} = \tiltraw{\dir}{\dir'}{K_0 \cup V'} =
\tiltraw{\dir}{\dir'}{K_0} \approx (0.8929, 0.4465, 0.0581)
\end{align*}
    (\algref{tilt}, see
\secref{tiltalgo} for details).
This $\dir_{\sigma_1}$ is then the output of
$\pop{K_0}{V_1}{W_1}{\dir}$, and a direction guaranteed to be~$K_0$-perpendicular to
$\sigma_1$.

\subsection{Simplex-Isolating Directions}\label{ssec:simplexisoexample}
Next, we find the directions of \substepref{higher}{pivot} in \conref{explicit}
(directions satisfying
\stmtref{prop:simplexiso-tiltdown} of \defref{simplexiso}). That is, for every
nonempty proper face of $\sigma_1$ and $\sigma_2$, we need a direction that pops
the face above the remaining vertices of the maximal simplex.
We walk through the procedure for a single face of~$\sigma_2$; the remaining
computations are similar and are summarized in \tabref{example}.

Consider the vertex $[v_0] \prec \sigma_2$. By \thmref{explicit}, we can 
use
$\pop{K_0}{\sigma_2}{v_0}{\dir_{\sigma_2}}$
(\algref{pivot}) to find a direction that pops~$v_0$ \mbox{above~$\sigma_2 \setminus
v_0 = [v_1,v_2]$}. First, in \alglnref{tiltdown-planefill}, we find $V' =
\planefill{\sigma_2}{\dir_{\sigma_2}}$ using \algref{planefill}.
Because~$\sigma_2$
has codimension one with~$\R^3$, this is simply $V' = \emptyset$.
Next, on \alglnref{tiltdown-findperp} of \algref{pivot}  we compute
\begin{align*}
    \dir' &=\perpendicular{{V' \cup (\sigma_2 \setminus v_0) \cup \{v_2 -
\dir_{\sigma_2}\}}} \\
    &= \perpendicular{{\{v_1,v_2\} \cup \{v_0 -
\dir_{\sigma_2}\}}} \\
    &\approx (0.9636, 0.1890, 0.1890).
\end{align*}
Because~$v_0$ is
above~$[v_1,v_2]$ with respect to $\dir'$, we do not enter the if loop on
\alglnref{tiltdown-ifstmt}. Finally, we use \algref{tilt} and compute
\begin{align*}
\dir =
\tiltraw{\dir_{\sigma_2}}{\dir'}{K_0} \approx (0.6598, 0.4944,
0.4944),
\end{align*}
    which is therefore also the output to
$\pop{K_0}{\sigma_2}{v_0}{\dir_{\sigma_2}}$.

\paragraph{Summary of Computations}
Below is the complete list of computed directions, including those that were
discussed in detail previously in this section.

\begin{table}[h]
            \caption{Directions
            that are vertex- and simplex-isolating for the simplicial complex of \figref{example},
            computed as described in \conref{explicit}}\label{tab:example}
            \begin{tabular}{@{}lll@{}}
                \toprule
                Property & Description & Computed Directions \\ && (rounded to four
                decimal places)\\
                \midrule
         \defref{vertiso}
          & vertex-isolating & (1,0,0) \\
          && (0,1,0) \\
          && (0,0,1) \\
          && (0.8013, 0.1603, 0.0385) \\
          \hline
                \stmtref{prop:simplexiso-perp} &
                $K_0$-perpendicular to $\sigma_2$  & $\dir_{\sigma_2} = (0.5774,
                0.5774, 0.5774)$
                \\ of \defref{simplexiso}
                & $K_0$-perpendicular to $\sigma_1$ & $\dir_{\sigma_1} =
                (0.8929, 0.4465, 0.0581)$\\
                \hline
                \stmtref{prop:simplexiso-tiltdown} & $(K_0, \sigma_2, [v_0],
                \dir_{\sigma_2})$-perturbation & (0.6598, 0.4944, 0.4944) \\
                of \defref{simplexiso} & $(K_0, \sigma_2, [v_1],
                \dir_{\sigma_2})$-perturbation & (0.5357, 0.6187, 0.5357)  \\
        & $(K_0, \sigma_2, [v_2], \dir_{\sigma_2})$-perturbation & (0.5357,
        0.5357, 0.6187) \\
        & $(K_0, \sigma_2, [v_0,v_1], \dir_{\sigma_2})$-perturbation & (0.5953,
        0.5953, 0.4866) \\
        & $(K_0, \sigma_2, [v_0,v_2], \dir_{\sigma_2})$-perturbation & (0.5959,
        0.4835, 0.5959)  \\
        & $(K_0, \sigma_2, [v_1,v_2], \dir_{\sigma_2})$-perturbation & (0.5235,
        0.5880, 0.5880) \\
        & $(K_0, \sigma_1, [v_0], \dir_{\sigma_1})$-perturbation &  (0.8879,
                0.4351, 0.0753)\\
        & $(K_0, \sigma_1, [v_3], \dir_{\sigma_1})$-perturbation & (0.8834,
                0.4504, 0.0579)  \\
            \botrule
            \end{tabular}
\end{table}